\newif\ifsubmission\submissionfalse
\newif\ifanon\anonfalse
\newif\ifqf\qffalse
\newif\ifcomments\commentstrue

\ifsubmission
  \documentclass[letterpaper,twocolumn,10pt]{article}
  \usepackage{usenix}
    \usepackage{mdwlist}

  \definecolor{refColor}{rgb}{1,0,0}
  \PassOptionsToPackage{colorlinks=true,linkcolor=refColor,citecolor=refColor}{hyperref}
  \usepackage{enumitem}
  \setlist[itemize]{noitemsep, topsep=0pt, leftmargin=.2in}
  \setlist[enumerate]{noitemsep, topsep=0pt, leftmargin=.2in}
  \setlist[description]{noitemsep,topsep=0pt,leftmargin=\parindent,labelindent=0pt}
  \newcommand{\subh}[1]{\par \vspace{2pt} \noindent \textbf{{#1}}}
  \renewcommand{\paragraph}[1]{\subh{#1}}
  \newcommand{\itparagraph}[1]{\par \vspace{2pt} \noindent \textit{{#1}}}
  \newenvironment{proof}{{\bf Proof:}}{\hfill\rule{2mm}{2mm}}
  \newenvironment{proofsketch}{{\bf Proof (sketch):}}{\hfill\rule{2mm}{2mm}}
\else
  \documentclass[letterpaper,11pt]{article}
  \usepackage{fullpage}
  \usepackage[hyphenbreaks]{breakurl}
  \PassOptionsToPackage{hyphens}{url}
  \usepackage[dvipsnames]{xcolor}
  \definecolor{refColor}{rgb}{0,0.2,0.5}
  \usepackage[colorlinks=true,linkcolor=refColor,citecolor=refColor]{hyperref}
  \usepackage{enumitem}
  \setlist[itemize]{itemsep=2pt, topsep=0pt, leftmargin=.5in}
  \setlist[enumerate]{itemsep=2pt, topsep=0pt, leftmargin=.5in}
  \setlist[description]{itemsep=2pt,topsep=0pt,leftmargin=\parindent,labelindent=0pt}
  \newcommand{\subh}[1]{\par \medskip \noindent \textbf{{#1}}}
  \renewcommand{\paragraph}[1]{\subh{#1}}
  \newcommand{\itparagraph}[1]{\par \medskip \noindent \textit{{#1}}}
  \usepackage{amsthm}
  
\fi

\usepackage{amsmath}
\usepackage{amssymb}

\usepackage{booktabs}
\usepackage{subcaption}

\usepackage{bm}
\usepackage{mathtools}
\usepackage{centernot}
\usepackage{stmaryrd}
\usepackage{adjustbox}
\usepackage{tabularx}

\makeatletter
\newcommand{\xMapsto}[2][]{\ext@arrow 0599{\Mapstofill@}{#1}{#2}}
\def\Mapstofill@{\arrowfill@{\Mapstochar\Relbar}\Relbar\Rightarrow}
\makeatother

\usepackage{framed}
\usepackage{multirow}

\usepackage{hyperref}
\usepackage{breakurl}
\makeatletter \g@addto@macro{\UrlBreaks}{\do\/\do\-} \makeatother  

\usepackage{amsmath}
\usepackage{amsfonts}
\usepackage{amssymb}
\usepackage{algorithm}
\usepackage{algpseudocode}
\algnewcommand\algorithmicinput{\textbf{Input:}}
\algnewcommand\Input{\item[\algorithmicinput]}
\algnewcommand\algorithmicrand{\textbf{Randomness:}}
\algnewcommand\Rand{\item[\algorithmicrand]}
\algnewcommand\algorithmicparam{\textbf{Parameters:}}
\algnewcommand\Param{\item[\algorithmicparam]}

\usepackage{xspace}

\newtheorem{theorem}{Theorem}
\newtheorem{lemma}{Lemma}

\newtheorem{definition}{Definition}

\newtheorem{threatmodel}{Threat Model}
\newtheorem{construction}{Construction}

\ifsubmission\else \theoremstyle{remark} \fi
\newtheorem{remark}{Remark}

\usepackage{tikz}
\usepackage{pict2e}        
\usepackage{pgfplots, pgfplotstable}
\usetikzlibrary{positioning}
\usetikzlibrary{decorations, shapes, arrows, calc}
\usetikzlibrary{decorations.text, decorations.markings}
\usetikzlibrary{fit}

\newcommand*{\StrikeThruDistance}{0.1cm}%

\usepackage{graphicx}
\graphicspath{ {figures/} }
\DeclareGraphicsExtensions{.pdf,.jpeg,.png}

\ifcomments
  \newcommand{\cmt}[1]{{\color{red}#1}}

  \newcommand{\sunoo}[1]{{\textcolor{violet}{/* SP: #1 */}}}
  \newcommand{\mike}[1]{{\textcolor{blue}{/* MS: #1 */}}}
  \newcommand{\matt}[1]{{\textcolor{olive}{/* MG: #1 */}}}
  
\else
  \newcommand{\cmt}[1]{\ignorespaces}

  \newcommand{\sunoo}[1]{\ignorespaces}
  \newcommand{\mike}[1]{\ignorespaces}
  \newcommand{\matt}[1]{\ignorespaces}
  
\fi

\newcommand{\powset}[1]{\mathbb{P}(#1)}

\newcommand{\eps}{\varepsilon}
\newcommand{\cind}{\approx_{c}}

\newcommand{\DE}{\hat{\Delta}}
\newcommand{\CurrTime}{\mathsf{CurrentTime}}

\newcommand{\KF}{KeyForge\xspace}
\newcommand{\BasicKF}{\KF}

\newcommand{\TF}{TimeForge\xspace}

\newcommand{\Sig}{\Sigma}
\newcommand{\Setup}{\mathsf{Setup}}
\newcommand{\KeyGen}{\mathsf{KeyGen}}
\newcommand{\KeyGenFromRoot}{\mathsf{KeyGen}^\star}
\newcommand{\Sign}{\mathsf{Sign}}
\newcommand{\Verify}{\mathsf{Verify}}
\newcommand{\Expire}{\mathsf{Expire}}
\newcommand{\Forge}{\mathsf{Forge}}

\newcommand{\Compress}{\mathsf{Compress}}

\newcommand{\HIBS}{\mathsf{HIBS}}
\newcommand{\FFS}{\mathsf{FFS}}
\newcommand{\SFFS}{\mathsf{BasicFFS}}

\newcommand{\RO}{\mathcal{O}}
\newcommand{\OSign}{\mathsf{S}}
\newcommand{\OExpire}{\mathsf{E}}

\renewcommand{\sec}{\kappa}
\newcommand{\NN}{\mathbb{N}}
\newcommand{\Msg}{\mathcal{M}}
\newcommand{\Tag}{\mathcal{T}}
\newcommand{\Id}{\mathcal{I}}
\newcommand{\Adv}{\mathcal{A}}
\newcommand{\Email}{\mathsf{E}}

\newcommand{\cA}{{\cal A}}

\newcommand{\cD}{{\cal D}}

\newcommand{\cS}{{\cal S}}

\newcommand{\y}{\mathsf{y}}
\newcommand{\m}{\mathsf{m}}
\renewcommand{\d}{\mathsf{d}}
\renewcommand{\c}{\mathsf{c}}

\newtheorem{innercustomthm}{Theorem}
\newenvironment{customthm}[1]
  {\renewcommand\theinnercustomthm{#1}\innercustomthm}
  {\endinnercustomthm}

\begin{document}

\newcommand{\papertitle}{\KF: Mitigating Email Breaches \\ with Forward-Forgeable Signatures}

\ifsubmission
\else
  \title{\papertitle}
  \author{
    Michael Specter \\
    MIT
  \and
    Sunoo Park\\
    MIT \& Harvard
  \and
    Matthew Green\\
    Johns Hopkins University
  }
\fi

\date{}

\maketitle


\begin{abstract}
Email breaches are commonplace, 
and they expose a wealth of personal, business, and political data
that may have devastating consequences.
The current email system allows any attacker who gains access to your email 
to prove the authenticity of the stolen messages to third parties
--- a property arising from a necessary anti-spam / anti-spoofing protocol called DKIM.
This exacerbates the problem of email breaches by
greatly increasing the potential for attackers to damage the users' reputation,
blackmail them, or sell the stolen information to third parties.

In this paper, we introduce \emph{non-attributable email}, which
guarantees that a wide class of adversaries are unable to 
convince any third party of the authenticity of stolen emails.
We formally define non-attributability, and 
present two practical system proposals --- \KF and \TF ~---
that provably achieve non-attributability
while maintaining the important protection against spam and spoofing
that is currently provided by DKIM.
Moreover, we implement \KF and demonstrate that that scheme is practical, 
achieving competitive verification and signing speed while also requiring 42\% \emph{less} bandwidth per email than RSA2048.
\end{abstract}

\section{Introduction}
\label{sec:intro}

Email is the world's largest and most ubiquitous messaging scheme, 
spanning industry, government, and personal use.
Email content comprises a vast trove of sensitive information; 
it is common knowledge that
email breaches can have devastating consequences.

Making matters worse, email protocols today ensure that if a malicious party gains access to your email, 
then she can cryptographically prove the authenticity of the stolen emails to any third party.
This greatly increases the intelligence value of stolen email:
attackers can (even anonymously) damage the reputation of users 
by publicly disseminating sensitive messages, or sell the messages to interested third parties, all with high credibility.
This property of email (``attributability'') arose as an unintended side-effect of 
the DomainKeys Identified Mail (DKIM) standard,\footnote{%
  As noted by a DKIM RFC author \cite{CallasDkimEmail}.
}
whose purpose is to fight spam and spoofing attacks.

The consequences of this design are far from hypothetical. For example, during the 2016 U.S. election season, leaked emails
from the Democratic National Committee (DNC) and other operatives were released to the public 
via Wikileaks~\cite{wikileaks_DNC}. While the original senders did not intentionally authenticate the emails, third parties
were able to verify the authenticity of emails due to the presence of cryptographic signatures from Google's email server~\cite{wikileaks_wikileaks_2016}. In some cases this authentication directly contradicted 
disavowals by party officials~\cite{brazileDenial}.

Regardless of whether one believes the 2016 DNC leaks were ultimately for the better or worse, it is alarming that an unintended side effect of a widely used protocol has created a broader ecosystem for credible propagation and misuse of illicitly obtained private communication. That DKIM materially supports such an ecosystem is clear.
Wikileaks has a webpage about DKIM authentication \cite{wikileaks_wikileaks_2016} indicating that the organization relies heavily on non-repudiability to enhance the credibility of its material. This is corroborated by recent incidents, e.g., \cite{fakedata}, in which email leaks were allegedly salted with falsified data. 

This raises a natural question: can we mitigate the potential harms of such an ecosystem, while maintaining the efficiency and strong spam- and spoofing-resistance of email today? It is a particularly pressing question as new reports of 
data breaches, including email account information, seem to surface every few weeks.
The 2013 breach of every one of Yahoo!'s 3 million email accounts \cite{yahoo} is a vivid reminder that
not only public figures are impacted by the consequences of compromised email.
Though the Yahoo! data has not been publicized,
data from other large breaches has been posted publicly online.\footnote{
E.g., the recent ``Collection 1'' breach implicating 773 million email addresses \cite{collection1}. Note that the ``Collection 1'' breach consists of email addresses and passwords used to log in to websites, not necessarily email accounts.}
Besides election season strategy,
attackers' motives have ranged from financial gain --- such as by selling patient healthcare data
gleaned from emails \cite{healthcare-data} --- to industrial espionage
and monitoring political dissidents and foreign officials \cite{csis}.
In the absence of any outside breach, such incentives may also prompt malicious insiders to sell or leak email data.

\bigskip
 
While the property that stolen email can be cryptographically attributed 
exacerbates email leaks, DKIM itself serves a necessary purpose in the email ecosystem. 
Concretely, DKIM requires email servers to provide a cryptographic signature
on each outgoing email's contents and metadata (including the sender's identity).
This assures the receiving server that the sending domain actually did intend to send the message, 
providing inter-domain accountability in case of spam, and moreover ensuring that spoofing is detectable.

An initial intuition may be that attributability of stolen email
is a necessary side-effect of any spam- and spoofing-resistant email system. Indeed, it is unclear how
a recipient can be certain that an email was legitimately sent
without gaining the ability to convince a third party of the same.

In this work, we challenge the above intuition and construct efficient protocols that
achieve the important spam- and spoofing-resistance
that DKIM provides, while simultaneously \emph{guaranteeing non-attributability} of stolen email.
Our protocols provably ensure that a wide class of adversaries are unable to 
convince any scrupulous third party of the authenticity of stolen emails.

\subsection{System Requirements}

Requirements on email demand certain unique properties not commonly seen in other messaging systems. For example, email is an any-mesh ecosystem; any domain owner on the internet could set up the appropriate DNS records and become an server that interoperates with any other email server. Similarly, a domain may be responsible for signing and verifying anything from hundreds to millions of emails per day, and administrators are likely to have to contend with increasing throughput requirements over time. 

In considering throughput for any domain, while maintaining good constants on computation time for signing and verification of emails is important, it is more important that the service be able to scale --- that adding more resources to the system actually provide linear or better performance guarantees. 
Scalability in terms in terms of interconnection with other mail servers is as important as email throughput.\footnote{The IETF standard for DMARC~\cite{kucherawy_domain-based_nodate} explicitly states that pre-sending agreements is a poor scalability choice for this reason, see also~\cite{thomas_dkim_requirements_nodate}.} 

An unintuitive result of these throughput requirements is that certain types of overhead that would be trivial in other messaging contexts, such as requiring servers to communicate prior to a message being sent or requiring round trips between servers on every message, are unlikely to be viable in this environment. For example, it would be difficult to scale a service that requires constant connection to, and maintain constant state with, an unbounded number of potentially malicious servers.

In addition to efficiency and scalability,
a practical solution for email non-attributability must address the following considerations:

\begin{description}
\item {\bf Untrusted recipients.}
A naive proposal is to rely on recipient behavior to ensure non-attributability:
e.g., having the receiving server delete all DKIM header information upon receipt.
Such solutions are inadequate for a threat model whose very purpose is
to prevent attacks by the recipient and/or the recipient's server (or an attacker who compromises the recipient).
Ideally, non-attributability for a given sender's emails should arise from
the sender's behavior alone. 

\item {\bf Asynchrony.}
Email is inherently asynchronous and must operate across many nodes that are unable to interact directly,
as discussed above.
Thus, schemes based on interactive authentication protocols, such as naturally deniable interactive zero-knowledge proofs
--- while they might suffice for a synchronous communication setting
--- are not viable solutions for email. 
Schemes based on other interactive protocols such as OTR~\cite{borisovOTR} are inappropriate for similar reasons.

\item {{\bf Universal forgeability}
One intuitive definition of email non-attributability
could be that DKIM signatures must be \emph{forgeable} by the recipient of an email
(``recipient forgeability'').
Then, even if the attacker presents validly signed emails a third party,
she lacks credibility: had she really broken into the recipient's email,
she could have produced the signatures herself.
Recipient forgeability could be achieved,
for example, by using replacing DKIM signatures with ring signatures \cite{RST01}
such that a valid signature can be produced by either the sender or the recipient,
as proposed by \cite{AHR05}.

\smallskip
However, recipient forgeability is a weak notion of non-attributability.
A stronger, and more desirable, notion is that DKIM signatures
must be forgeable by \emph{anyone} after a small time delay to prevent spam and spoofing attacks.
This renders alleged attackers far less credible,
since anyone can forge valid-looking signatures without breaking into any email account at all.
We call this stronger notion ``universal forgeability,''
and consider it an essential feature of a satisfactory
non-attributability solution for email.}

\item{\bf Incremental deployment.}
Given the vast number of existing email servers and the need for interoperability, we consider the majority of the email ecosystem to be entrenched. For example, it would be difficult to require substantial changes to mail routing, and one cannot expect that every actor immediately switch to some new scheme. We do, however, hold the view that DKIM itself can be replaced incrementally, and that new signing mechanisms represent a less onerous impact on the email ecosystem.

\item {\bf Long-lived keys.}
One natural approach to short-lived credibility of signatures
is to leverage correspondingly short-lived keys and publish each secret key
at the end of its lifetime; this sort of approach to achieving forgeability
has been mentioned in passing outside of the context of email \cite{borisovOTR}.
Similarly, one may leverage easily breakable keys in tandem with rapid key rotation, a solution originally suggested by Jon Callas, one of the original creators of the DKIM standard~\cite{CallasDkimEmail}.

\smallskip
Unfortunately, too-frequent key rotation entails a variety of likely insurmountable practical problems;
rotating keys is a manual process that introduces risk of misconfiguration and key theft, and no administrator wants to be responsible for disabling email for an entire organization. 
Worse, DNS results are often cached, so replacing an individual DNS record with a new one takes time and can yield inconsistent results for different DNS servers on the receiver's end. 
Thus, given the impracticability of very frequent key rotation,
this simple scheme cannot give a satisfactory non-attributability guarantee.

\end{description}

These observations point to a possible solution involving long-lived \emph{public} keys
but short-lived \emph{secret} keys, reminiscent of
forward-secure signatures (FSS) \cite{And97,BM99}.
However, FSS were designed with a rather different goal from ours: namely,
to allow efficient key updating while preventing derivation of \emph{past} keys from present
and future keys.
In contrast, our setting requires that
\emph{present and future} keys cannot be derived from past keys.
The only way to achieve this property using an FFS
would be to precompute a long list of keys and then use the secret keys \emph{in backwards order} ---
this is arguably better than the simplest solution based on short-lived keys,
but having to store the whole list of keys is inefficient and unsatisfactory.

Our solution, \KF, does in fact leverage a signature scheme with 
long-lived public keys and short-lived secret keys:
we define a new primitive called \emph{forward-forgeable signatures} (FFS)
capturing the specific requirements of our application,
and give a construction of FFS from any hierarchical identity-based signature scheme,
as detailed in Section~\ref{sec:ffs}. Our second proposal, \TF, requires no exposure of private keys, but achieves a slightly weaker notion of non-attributability than \KF. \TF is explained in Section~\ref{sec:protocols:tf}.

\subsection{Key Ideas}\label{sec:intro:keyideas}

There are two main ideas underlying the solutions we explore:
\emph{delayed universal forgeability} and \emph{immediate recipient forgeability}.

\paragraph{Delayed universal forgeability.}
  The key idea in this approach is to ensure that signatures with respect to past emails ``expire'' after a time delay $\Delta$ and
  thereafter become forgeable by the general public (i.e., arbitrary outsiders or non-parties).
  This property ensures that no attribution will be credible after the time delay has elapsed.
  We call this property \emph{delayed universal forgeability}.
  As long as $\Delta$ is set larger than the maximum viable time for email latency,
  the signature will still be convincing to the recipient at the time of receipt,
  thus maintaining the spam- and spoofing-resistance of DKIM.

  Signatures that possess delayed universal forgeability retain all of the standard unforgeability properties
  one would expect of a standard signature scheme, until the set time $\Delta$ has passed.
  Thus in cases where an attacker gains access to email and shows it to a third party within $\Delta$ time
  after the email was sent, a third party will be convinced of the email's authenticity. 
  Effectively, delayed universal forgeability protects against adversaries that compromise
  an email account by breaking in and taking a snapshot (``after-the-fact attacks''), but not
  adversaries that fully control an email account and monitor its email in real time (``real-time attacks'').
  After-the-fact attacks cover a broad range of realistic attacks, for example, including many data breaches.
  Next, we discuss how we address real-time attacks.

\paragraph{Immediate recipient forgeability.}
  Suppose that the fact of access to a particular client account
  implies the ability to forge messages from arbitrary other servers \emph{to that recipient only}:
  that is, the ability to obtain valid DKIM signatures 
  on email content and metadata of one's choice.
  We call this \emph{immediate recipient forgeability}.
  Importantly, the recipient constraint ensures the
  inability to impersonate any other server for the purposes of email addressed to \emph{other} recipients,
  thus not interfering with DKIM's spam- and spoofing-resistance.

  This undermines the credibility of any attacker who claims to 
  have ongoing access to a particular email account
  and attempts to convince third parties of the authenticity of emails 
  that she claims were sent to (and from) that account.
  This is also the case for real-time attacks, which may publish
  allegedly-incoming emails immediately after they are received.

\paragraph{Combining both notions.}
  While immediate recipient forgeability has the advantage of no time delay,
  it has the disadvantage that the ability to forge is not universal: rather,
  it is limited to those who have access to the recipient email account.
  In contrast, delayed universal forgeability enables anyone at all to forge.
  Our protocols achieve the ``best of both worlds,'' in order to leverage both their advantages.

  Section~\ref{sec:model} defines our threat model,
  discusses its limitations, and formalizes immediate recipient forgeability
  and delayed universal forgeability 
  with indistinguishability-based definitions (Definitions~\ref{def:NA1}--\ref{def:NA2}).

\subsection{Our Solutions}\label{sec:intro:solns}

This paper fully specifies and evaluates our main proposal, \KF.
We also outline an alternative protocol \TF
that achieves email non-attributability by rather different techniques
assuming the existence of a timestamping authority.

\paragraph{\KF}
Our main proposal is \KF, a protocol to replace DKIM and the associated authentication mechanisms.
The first idea of \KF is to achieve \emph{delayed universal forgeability}
by publicizing signing keys after a set time delay $\Delta$ (say, by posting them on the internet).

We present a new definition, \emph{forward-forgeable signatures (FFS)}, 
to formalize the requirements of a signature scheme equipped with a method to
selectively disclose ``forging information'' for subsets of past signatures.
The public key is required to persist over disclosures:
otherwise, the frequency of manual key updates and DNS updates would render fine-grained disclosures inviable.
We moreover define \emph{succinctness} of FFS, a measure of how efficiently the disclosures can be made.
We construct FFS based on hierarchical identity-based signatures (HIBS)
and use it as an essential building block for \KF.
Our construction achieves logarithmic succinctness.

The second idea of \KF is to achieve \emph{immediate recipient forgeability}
by having email servers accept ``forge requests''
requesting that specific emails be sent to the requester in real time.
The restriction that the recipient must be the requester
ensures that DKIM's spam- and spoofing-protection remains effective.

\paragraph{\TF.}
Our second proposal \TF 
assumes a \emph{publicly verifiable timekeeper (PVTK)} model, in which 
a reliable timekeeper periodically issues timestamps
whose authenticity is publicly verifiable (e.g., because the authority signs each timestamp).
Rather than deploying and scaling a new timekeeping service, we show
that viable timekeeping systems can be constructed from existing Internet 
infrastructure, including OCSP staple responders~\cite{rfc6961}, Certificate Transparency~\cite{rfc6962}, randomness beacons~\cite{nistbeacon}, and proof-of-work blockchains~\cite{nakamoto2012bitcoin,wood2014ethereum}.

In a nutshell, the idea of \TF is to substitute each DKIM signature on a message $m$ at time $t$
with a succinct zero-knowledge proof of the statement $S(m)\vee R(m)\vee T(t+\Delta)$,
where: $S(m)$ denotes knowledge of a valid signature by the sender on $m$,
$R(m)$ denotes knowledge of a valid signature by the receiver on $m$,
and $T(t+\Delta)$ denotes knowledge of a valid timestamp for a time later than $t+\Delta$.
Effectively, the inclusion of $T(t+\Delta)$ ensures \emph{delayed universal
 forgeability},
and the inclusion of $R(m)$ ensures \emph{immediate recipient forgeability}.
The primitive just described can be thought of as a
forward-forgeable signature scheme in the PVTK model.

\subsection{Summary of Contributions}

\begin{itemize}
  \item We provide formal \emph{definitions} of:
    \begin{itemize}
      \item \emph{Email non-attributability} (Section~\ref{sec:sysreq:def})
      \item \emph{Forward-forgeable signatures} (Section~\ref{sec:ffs})
    \end{itemize}
  \item We give \emph{provably secure constructions} of succinct forward-forgeable signatures (FFS)
  in both the standard and PVTK models.
  FFS may be of independent interest beyond the present application.
  \item We present two \emph{system designs} --- \KF, and \TF~--- that achieve email non-attributability
    and could be applied to the existing email system with acceptable practical and processing overhead.
  \item We present a fully \emph{implemented} \KF and provide a comprehensive evaluation of \KF's impact on signing, verification, and bandwidth costs.
\end{itemize}

\section{Background on Email}
\label{sec:bg}

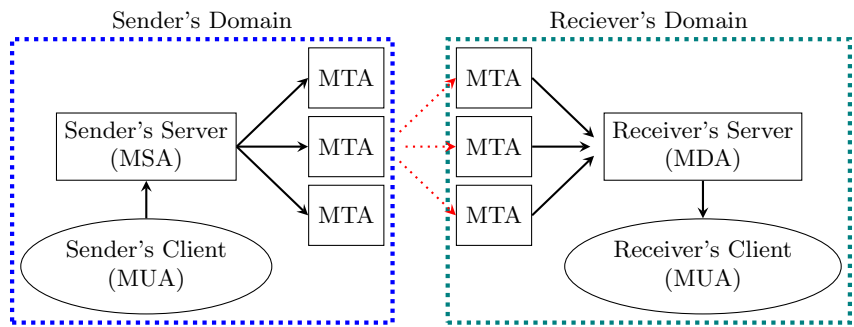
\begin{figure}[!ht]
    \centering
    {\footnotesize
    \begin{tikzpicture}
    \tikzstyle{server} = [rectangle, draw, minimum width=1cm, minimum height=.8cm, align=center]
    \tikzstyle{client} = [ellipse, draw, shape border rotate=90, align=center, aspect=.3, minimum width=1.1cm,style={font=\footnotesize}]
    \tikzset{mytext/.style={sloped,anchor=center,below=3pt, font=\footnotesize}, align=center}
    \tikzset{mytext2/.style={right=3pt, font=\footnotesize}, align=center}
    \tikzset{mytext1/.style={left=3pt, font=\footnotesize}, align=center}
    \tikzset{myarr/.style={>=stealth,thick}}
    \tikzset{myarr2/.style={>=stealth,thick, dotted,draw=red}}
    
    \node[server] (intermeda1) {MTA};
    \node[server] (intermeda2) [below=.1cm of intermeda1] {MTA};
    \node[server] (intermeda3) [below=.1cm of intermeda2] {MTA};

    \node[server] (intermedb1) [right=.95cm of intermeda1] {MTA};
    \node[server] (intermedb2) [below=.1cm of intermedb1] {MTA};
    \node[server] (intermedb3) [below=.1cm of intermedb2] {MTA};

    \node[server] (aserver) [left=.95cm of intermeda2] {Sender's Server\\(MSA)};
    \node[client] (aclient) [below=.5cm of aserver] {Sender's Client\\(MUA)};

    \node[server] (bserver) [right=.95cm of intermedb2] {Receiver's Server\\(MDA)};
    \node[client] (bclient) [below=.5cm of bserver] {Receiver's Client\\(MUA)};

    \node[draw,dotted,ultra thick,draw=blue,fit=(intermeda1) (intermeda3) (aserver) (aclient),label=Sender's Domain] (adomain) {};
    \node[draw,dotted,ultra thick,draw=teal,fit=(intermedb1) (intermedb3) (bserver) (bclient),label=Reciever's Domain] (bdomain) {};

    \draw[->,myarr] (aserver.east) -- (intermeda1.west);
    \draw[->,myarr] (aserver.east) -- (intermeda2.west);
    \draw[->,myarr] (aserver.east) -- (intermeda3.west);

    \draw[->,myarr2,shorten <=8pt] (intermeda2.east) -- (intermedb1.west);
    \draw[->,myarr2,shorten <=8pt] (intermeda2.east) -- (intermedb2.west);
    \draw[->,myarr2,shorten <=8pt] (intermeda2.east) -- (intermedb3.west);

    \draw[->,myarr,shorten >=5pt] (intermedb1.east) -- (bserver.west) ;
    \draw[->,myarr,shorten >=5pt] (intermedb2.east) -- (bserver.west) ;
    \draw[->,myarr,shorten >=5pt] (intermedb3.east) -- (bserver.west) ;

    \draw[->,myarr] (aclient) -- (aserver);
    \draw[->,myarr] (bserver) -- (bclient);

    \end{tikzpicture}
    \caption{Simplified email routing infrastructure}
    \label{fig:MTA}
    }
\end{figure}

RFC 5598 \cite{RFC5598} defines terminology for processing and delivering mail. 
This section introduces basic terminology towards discussing our threat model and system requirements. 

\subsection{Email Routing}

Figure~\ref{fig:MTA} illustrates the architecture of email routing: 
an asynchronous routing protocol built on top of TCP/IP. 

Users first establish a relationship with a trusted email service provider, 
called a Mail Submission Agent (MSA) on the sender side
and a Mail Delivery Agent (MDA) on the receiver side. 
The MDA is responsible for 
various tasks including verifying the 
authenticity of incoming messages. 

The user's email client is generically called a Mail User Agent (MUA). 
Email originates from an MUA, and arrives at the user's trusted MSA. Depending on the system's configuration, the MSA may send the email to other Mail Transfer Agents (MTAs) it trusts. Eventually, an MTA performs a DNS lookup for the receiving domain
to check which MTAs are authorized to process emails for that domain.
The email is then sent to one of these MTAs. 
After a number of hops that depend on the sending and receiving organizations' infrastructure, the email 
arrives at the receiver's MDA, which is responsible for verifying the message for the receiver's MUA.  

\subsection{Email Authentication}
The IETF has developed a number of standards that allow domains to sign and verify incoming and outgoing messages. We provide a 
stylized overview of the four that have seen appreciable adoption: DKIM, SPF, DMARC, and ARC. 

\begin{description}
\item[DKIM.] DomainKeys Identified Mail (DKIM) is an IETF standard for an MSA to sign outgoing email, and an MDA to verify that email by looking up the MSA's public key in the DNS. 
\begin{description}
\setlength{\itemsep}{5pt}
\item[\bf Setup:] The MSA generates a  
key pair and uploads the public key to the DNS.
\item[\bf Send:] The MSA adds the location of its public key to the email's metadata (or \emph{header}), as well as metadata allowing verification of the signature, and then signs the email and its headers.\footnote{This usually includes a hash of the entirety of the message, but the specification does allow for portions of the body of the message to go unsigned but it's unclear how often this is done, and appears to be relatively rare.} 
\item[\bf Verify:] On receipt, the MDA performs a DNS lookup for the sender's public key as specified in the header, and then cryptographically verifies the email. 
\end{description}

\item[SPF] 
The Sender Policy Framework (SPF) ensures that intermediary MTAs are permitted to send and receive messages as a part of the domain. In a way, this solves a somewhat orthogonal problem to DKIM in that SPF provides little guarantee that the message itself has not been modified by an intermediary, but instead provides spoofing protection by limiting what IP addresses are valid accepting MTAs.

\item[DMARC] 
An SPF or DKIM failure as a result of a misconfiguration is indistinguishable from a failure due to an attempted message spoofing, and neither DKIM nor SPF provide mechanisms for alerting the sending domain that there has been a problem. DMARC solves this by adding a DNS TXT record specifying to the receiver what it should do in the case of such failures (such as quarantine, reject, or accept the message despite the failure), as well as providing an email address to send aggregated statistics on such failures. 

\item[ARC]
Authenticated Received Chain (ARC) is an experimental and largely unimplemented IETF standard that aims to address the issues caused by indirect email flows. 

Legitimate modification of messages in transit may occur in a number of circumstances. For example, an MTA that is also a virus scanner may remove malicious attachments, or a mailing list may prepend the mailing list's name to the subject of an email. Unfortunately, any alteration of the body or headers invalidates the original DKIM signature.

ARC acts as an attestation by an intermediary that some subset of DKIM, SPF, DMARC, or previous ARC instances were verifiable before content was modified. To do so, the intermediary adds its own signature to the header of the message, along with metadata about what was verified. 

\end{description}

\subsection{Challenges for Synchronous Authentication Protocols}

Email is inherently asynchronous, and any replacement for DKIM must therefore allow for non-interactive verification of emails. The reason for this is twofold:

\paragraph{Third party MTAs.} Many organizations choose to use third-party MTAs as an initial hop between the organization's own mail server and the internet. The reason one may choose to use such MTAs is varied, but a common use case is to provide some protection from spam, malicious attachments, or DDoS attacks. While these intermediaries are allowed to quarantine messages or provide flow control to the MDA, they are not trusted to refrain from modifying or spoofing emails. By scraping DNS results for the Alexa top 150k, we find that 31,615 have an MX record, and, of these, 22\% (6,793) are using a confirmed multi-hop third-party MTA provider.

\paragraph{Mail forwarding.} Often users will forward email received from an account on one domain to another through mail forwarding. These emails then must be verified by the receiving domain, and there is no way for the receiver to communicate with the original sender.

\section{Model and Security Definitions}\label{sec:model}

This section presents threat models
and formal definitions of email non-attributability.

\paragraph{Notation}
``PPT'' means ``probabilistic polynomial time.''
$|S|$ denotes the size of a set $S$. 
$[n]$ denotes the set $\{1,\dots,n\}$ of positive integers up to $n$,
and $\powset{\cdot}$ denotes powerset.
$\cind$ denotes computational indistinguishability.
$\tau||e$ denotes the result of appending an additional element $e$
to a tuple $\tau$.

\subsection{Model}\label{sec:sysreq:model}

\paragraph{Time}
We model time in discrete time-steps and assume that
everyone has access to a fairly accurate local clock.
(This is a reasonable model because our
system requires synchrony only at a coarse granularity, such
that realistic discrepancies in local clocks will not have any effect.)

\paragraph{Asynchrony} Recognizing the asynchronous nature of email,
$\DE$ denotes an upper bound on the time required for email delivery.
Our parameter settings depend on $\DE$, and our evaluation sets $\DE$ at 15 minutes,
as discussed in Section~\ref{sec:protocols:kf:hibs-details}.

\paragraph{DNS}
Our model assumes all parties and algorithms
have access to DNS and can update their own DNS records.

\paragraph{Publication}
We assume that each party has method of publishing persistent, 
updatable information that is retrievable by all other parties and algorithms.
This could be via DNS or on another medium,
such as posting information on a website.

\subsection{Threat models}
\label{sec:sysreq:tm} \label{sec:model:tm}

The type of attack with which we are concerned is 
disclosure of private communications at the \emph{recipient's} server, whether because
the recipient is compromised or because it is malicious.
We define two threat models, defined below.

Both our proposals \KF and \TF achieve security against Threat Model~\ref{tm:maliciousOnR},
which is the stronger of the two threat models and captures advanced persistent threats.
Threat Model~\ref{tm:honestOnR}
may be a meaningful threat model in a scenario where attackers may illicitly
gain access to an email server but are not likely to be able to maintain 
such access for extended periods of time.
A basic, simpler version of \KF achieves non-attributability for almost all emails
(i.e., all but very recent emails)
under Threat Model~\ref{tm:honestOnR}.

\begin{threatmodel}\emph{(After-the-fact attacks)}\label{tm:honestOnR}
In this model, the recipient is presumed honest at the time of email receipt,
but is later compromised by an attacker that takes a snapshot of all past email content.
\end{threatmodel}

\begin{threatmodel}\emph{(Real-time attacks)}\label{tm:maliciousOnR}
In this model, the receiver is presumed malicious at the time of email receipt,
i.e., presumed to have ongoing access and immediate intent 
to disclose the content of received emails to third parties.
\end{threatmodel}

\paragraph{Client-server trust.}
In the current system, email clients rely heavily on their email servers.
A malicious email server
could easily and undetectably misbehave in many essential functions:
it could drop incoming emails, modify outgoing emails (since typically,
emails are not signed client-side),
or falsify the content and metadata of incoming emails (since typically,
clients do not perform DKIM verification themselves).

Thus, client-server trust is very high in practice.
We treat the client and server as a single entity for the purposes of this paper,
except where specifically otherwise indicated.

\paragraph{No trusted parties.}
In a system where credibility is based on a reliable reputation system,
a trusted (reputable) party with the ability to eavesdrop on the communication channel 
would be able to undermine non-attributability by keeping logs of all traffic it sees.
For example, an MTA in the email system, if it kept logs of all traffic,
could effectively resolve the question of authorship of emails by referring to its internal logs,
in a manner convincing to anyone who simply trusts the MTA based on its reputation.
Our model assumes mutually distrustful parties: i.e., that no party
is taken simply on its word as in the scenario just described.
In other words, credibility in our model is evidence-based and not reputation-based.

\subsection{Defining Non-Attributability}\label{sec:sysreq:def}\label{sec:model:def} 

We define email non-attributability as a game involving
an email protocol $\Email$, adversary $\cA$,
simulator $\cS$, and distinguisher $\cD$.
For any email server $S$ with internal state $s$,
we denote by $\Email_s(S,R,m,\mu,t)$
the information that recipient $R$ receives
when $S$ legitimately sends at time $t$ an email message $m$ with metadata $\mu$ to $R$.
For simplicity, our notation leaves implicit that parties other than the sender
participate in transmission and may affect the information received by $R$ (for example, the MTAs);
however, the final received information $\Email_s(S,R,m,\mu)$ should be thought
to contain any modifications made en route between $S$ and $R$.
While this definition refers to ``internal state $s$'' for generality,
the state $s$ can for our purposes be thought to contain secret key material.

Intuitively, we require indistinguishability between 
a legitimate email $\Email_s(S,R,m,\mu,t)$ and
a ``fake'' email that was created without access to the sending server at all:
that is, without knowing $s$.
To model this, we consider a simulator $\cS$ whose ``goal'' is 
to create an email that looks legitimate without knowing $s$,
and our security definition requires that $\Email_s(S,R,m,\mu,t)$ is indistinguishable 
from the output of $\cS$.

Next, we give two formal definitions of non-attributability.
The first (Definition~\ref{def:NA1}) considers
a simulator that has access to a particular recipient's email server,
and is required to output email from any sender to that recipient.
The second (Definition~\ref{def:NA2}) is a stronger definition
that considers a simulator that
is required to output email from any sender to any recipient
and has access to neither the sender's nor the receiver's email server.

In other words, if an adversary publishes an email allegedly authored by an honest party,
Definition~\ref{def:NA1} guarantees that she can credibly argue that,
granting that the attacker indeed broke into her or her correspondent's account,
his allegation inherently lacks credibility because 
he must have had the ability to forge emails with content of his choice 
(by virtue of having access to one of their email accounts).
Definition~\ref{def:NA2} gives the yet stronger guarantee that \emph{anyone}
can forge past emails: 
so the allegation is even less credible and the email accounts may not even have been compromised.
Note that the two definitions are incomparable: neither is strictly stronger.

\begin{definition}[Recipient non-attributability]\label{def:NA1}
An email protocol $\Email$ is \emph{non-attributable for recipient servers} if
there is a PPT simulator $\cS$ such that for any sender $S$ and recipient $R$,
for any email message $m$ and metadata $\mu$,
$$\Email_s(S,R,m,\mu,t) \cind \cS(R,S,m,\mu) \ ,$$
where $s$ is the internal state of $S$, $t$ is the time at which $\cS$ is invoked, and
$R$ denotes the internal state of the recipient $R$'s email server.

$\Email$ is furthermore \emph{non-attributable for recipients} if
there is a PPT simulator $\cS$ such that for any sender $S$ and recipient $R$,
for any email message $m$ and metadata $\mu$,
$$\Email_s(S,R,m,\mu,t) \cind \cS^{R}(S,m,\mu) \ ,$$
where the superscript $R$ denotes that $\cS$ has the capability of sending outgoing
mail as $R$ through $R$'s email server.
\end{definition}

Observe that \emph{non-attributability for recipients}
is a stronger than (i.e., implies) \emph{non-attributability for recipient servers}.
\KF achieves non-attributability for recipients,
whereas \TF achieves non-attributability for recipient servers.
See Remark~\ref{rmk:RNA} for further discussion of the differences
between the two variants of recipient non-attributability.

\begin{definition}[$\Delta$-universal non-attributability]\label{def:NA2}
For $\Delta\in\NN$, an email protocol $\Email$ is \emph{$\Delta$-strongly non-attributable} if
there is a PPT simulator $\cS$ such that for any sender $S$ 
(with internal state $s$) and recipient $R$,
for any email message $m$, metadata $\mu$, and timestamp $t$,
the following holds at any time $\geq t+\Delta$:
$$\Email_s(S,R,m,\mu,t) \cind \cS(S,R,m,\mu,t) \ .$$
\end{definition}

The two indistinguishability requirements defined above
serve to ensure that no attacker can credibly claim to a third party\footnote{%
E.g., the general public (if the allegedly stolen emails are released publicly) 
or a specific interested party (such as a potential buyer or disseminator of the information).}
that he is providing her with the true content of an email account:
the third party is in the role of distinguisher.

Definition~\ref{def:NA2} is inviable if $\Delta<\DE$.
Otherwise, the spam- and spoofing-resistance provided by DKIM
would be undermined, since any outsider could use the simulator in real time
to send spam email indistinguishable to the recipient from email actually sent by an honest party.
This means that the simulator can only begin to generate an indistinguishable
distribution after $\DE$ delay.
Recall that as time elapses, new information may become available to the simulator,
as discussed in Section~\ref{sec:sysreq:model}.

\paragraph{Relation to the threat models.}
$\Delta$-universal non-attributability (Definition~\ref{def:NA2})
achieves non-attributability
against after-the-fact attacks (Threat Model~\ref{tm:honestOnR})
for all emails sent and received at least $\Delta$ before the server is compromised.

The combination of recipient non-attributability and $\DE$-universal non-attributability
(Definitions~\ref{def:NA1} and \ref{def:NA2})
yields effective non-attributability for all emails against real-time attacks
(Threat Model~\ref{tm:maliciousOnR}).
A real-time attacker with ongoing access to an email server
can easily make the fact of his access evident by immediately publishing
all emails he sees (in particular, within time $\DE$ of receipt),
but will nonetheless be unable to convince a third party of any particular emails
since the fact of his access to the email server allows him to forge emails in real time,
under Definition~\ref{def:NA1}.
With respect to allegedly compromised emails from more than $\DE$ ago,
an attacker's credibility is even more diminished, since for such past timestamps
\emph{anyone with internet access} can generate seemingly validly signed emails,
even without breaking into any email server at all, under Definition~\ref{def:NA2}.

\paragraph{Limitations of our model.}
A practical consequence of recipient non-attributability is that
a recipient $R$'s email server can, \emph{unknown to $R$}, create fraudulent messages that 
appear to be legitimate emails from any sender to $R$, and deliver them to $R$.
As discussed Section~\ref{sec:model:tm},
the level of client-server trust necessitated by the current email system is very high.
In this context, recipient non-attributability does not meaningfully increase 
the trust a client places in her email server.
For example, email servers in the current system could (and often do)
omit DKIM headers when delivering emails to clients:
this effectively implies the ability to deliver fake messages to their clients.

A trivial and uninteresting way to achieve Definitions~\ref{def:NA1} and \ref{def:NA2} is not to sign emails at all.
Of course, this is undesirable as it would undermine the spam- and spoofing-resistance
for which DKIM was designed. 
The concurrent requirement of spam- and spoofing-resistance is implicit
throughout this paper.
It follows that any simulator $\cS$ satisfying Definition~\ref{def:NA1}
must use the recipient $R$'s secret state $r$:  that is, in order to prevent spam,
real-time forgery
must be limited to messages whose recipient is the forger herself.
This suggests that recipient non-attributability would not be meaningful 
in an extreme situation where every single use of $r$ can be monitored and attested to,
since then an attacker could
prove that $\cS$ was never invoked on $r$. 
This might be plausible assuming secure hardware, e.g., by
generating and monitoring all uses of $r$ within a secure enclave (as suggested in~\cite{gunn2018circumventing}) ---
but even then, such an attack would likely only be feasible
by the unlikely attacker who has designed her recipient email server with this unlikely configuration from its very setup.
We consider such attacks to be outside our threat model.

Finally, we note that our definitions
allow for attackers to convince others of the fact that they have access to a particular email account.
Our guarantee is that even so, 
they cannot make credible claims about the content of the emails,
since they gain the ability to falsify emails by the very fact of their access.
That attackers can prove access is unavoidable given that
universal forgeability is not compatible with spam resistance for too small $\Delta$: e.g.,
consider a real-time attacker that obtains a secure timestamp (via an outside timestamping service) for all incoming emails immediately on arrival.
Under Definition~\ref{def:NA1}, this does \emph{not} prove authenticity of the email content,
but just demonstrates access to the recipient's account. In general, we treat such attacks as being outside our threat model.

\section{\mbox{Forward-Forgeable Signatures}}
\label{sec:ffs}

\subsection{Definition}

Definition~\ref{def:FFS} formalizes \emph{forward-forgeable signatures (FFS)}.
They are a new primitive that this paper introduces, and
are an essential building block of our main protocol \KF.
Informally, FFS are signature schemes equipped with a method to
selectively ``expire'' past signatures by releasing \emph{expiry information} that makes them forgeable.
More precisely: in an FFS, each signature is made with respect to a \emph{tag} $\tau$,
which can be an arbitrary string. 
Expiry information can be released with respect to any tag or set of tags.

In our context of wishing to make past signatures forgeable after a time delay,
the tag can be thought to be a timestamp: for example, each email is signed with respect to the current time $\tau$,
and at some later time $\tau+\Delta$, the signer may publish expiry information for $\tau$.

The \emph{correctness} requirement of FFS is the same as that of standard signatures.
The \emph{unforgeability} requirement is modified to include an \emph{expiry oracle}:
that is, unforgeability of signatures w.r.t. non-expired tags must hold
even in the presence arbitrary, adversarially chosen expirations.
The \emph{forgeability on expiry} requirement is a feature of FFS that has no analogue in standard signatures:

\begin{definition}[FFS]\label{def:succinctFDS}\label{def:FFS}
A \emph{forward-forgeable signature scheme (FFS)} $\Sig$
is implicitly parametrized by message space
$\Msg$ and tag space $\Tag$, and consists of five algorithms
$$\Sig=(\KeyGen,\Sign,\Verify,\Expire,\Forge)$$
satisfying the following syntax and requirements.

\smallskip
\noindent\textsc{Syntax:}
\begin{itemize}
  \item $\KeyGen(1^\sec)$ takes as input a security 
  parameter\footnote{Technically, all five algorithms take $1^\sec$ as an input,
  and $\Msg$ and $\Tag$ may be parametrized by $\sec$.
  For brevity, we leave this 
  implicit except in $\KeyGen$.} $1^\sec$ and outputs a key pair $(vk,sk)$.
  \item $\Sign(sk, \tau, m)$ takes as input a signing key $sk$, a tag $\tau\in\Tag$,
  and a message $m\in\Msg$, and outputs a signature $\sigma$.
  \item $\Verify(vk, \tau, m, \sigma)$ takes as input a verification key $vk$, a tag $\tau\in\Tag$,
  a message $m\in\Msg$, and a signature $\sigma$, and outputs a single bit indicating
  whether or not $\sigma$ is a valid signature with respect to $vk$, $m$, and $\tau$.
  \item $\Expire(sk, T)$ takes as input a signing key $sk$ and 
  a tag set $T\subseteq\Tag$, and outputs expiry info $\eta$.
  \item $\Forge(\eta, \tau, m)$ takes as input expiry info $\eta$, 
  a tag $\tau\in\Tag$, and a message $m\in\Msg$, and outputs signature $\sigma$.
\end{itemize}
\noindent\textsc{Required properties:}
\begin{enumerate}
  \item \emph{\underline{Correctness:}} 
    For all $m\in\Msg,\tau\in\Tag$,
    there is a negligible function $\eps$ such that for all $\sec$,
    \begin{equation*}
      \ifsubmission\else\footnotesize\fi
      \Pr\left[
        \begin{array}{lll}
          \begin{array}{l}
          (vk,sk)\gets\KeyGen(1^\sec) \\
          \sigma\gets\Sign(sk,\tau,m) \\
          b\gets\Verify\left(vk,\tau,m,\sigma\right)
          \end{array}
          &:&
          b=1
        \end{array}
      \right] \geq 1-\eps(\sec) \ .
    \end{equation*}
  \item \emph{\underline{Unforgeability:}} For any PPT $\Adv$,
  there is a negligible function $\eps$ such that for all $\sec\in\NN$,
  \begin{align*}
    \ifsubmission\else\footnotesize\setlength\arraycolsep{3pt}\fi
    \Pr\left[
      \begin{array}{lll}
        \begin{array}{l}
        (vk,sk)\gets\KeyGen(1^\sec) \\
        (\tau,m,\sigma)\gets\Adv^{\OSign_{sk},\OExpire_{sk}}(vk) \\
        b\gets\Verify\left(vk,\tau,m,\sigma\right) \\
        b' = \tau\notin Q'_{\OExpire} \wedge (\tau,m)\notin Q_{\OSign}
        \end{array}
        &:&
        \begin{array}{l}
        b=b'=1 \\
        \end{array}
      \end{array}
    \right] \\ \leq\eps(\sec) \ ,
  \end{align*}
  where $\OSign_{sk}$ and $\OExpire_{sk}$ respectively denote oracles
  $\Sign(sk,\cdot,\cdot)$ and $\Expire(sk,\cdot)$,
  $Q_{\OSign}$ and $Q_{\OExpire}$ denote the sets of queries made
  by $\Adv$ to the respective oracles, and $Q'_{\OExpire}=\bigcup_{T\in Q_{\OExpire}}T$.
  \item \emph{\underline{Forgeability on expiry:}}
  For all $m\in\Msg,T\subseteq\Tag$, for any $\tau\in T$, 
  there is a negligible function $\eps$ such that for all $\sec$,
  \begin{equation*}
    \ifsubmission\else\footnotesize\fi
    \Pr\left[
      \begin{array}{lll}
        \begin{array}{l}
        (vk,sk)\gets\KeyGen(1^\sec) \\
        \eta\gets\Expire(sk,T) \\
        \sigma\gets\Forge(\eta,\tau,m) \\
        b\gets\Verify\left(vk,\tau,m,\sigma\right)
        \end{array}
        &:&
        b=1
      \end{array}
    \right] \geq 1-\eps(\sec)\ .
  \end{equation*}
\end{enumerate}
\end{definition}

\paragraph{Generalized expiry policies}
Definition~\ref{def:FFS} allows for more complex \emph{expiry policies} than
simply expiring a sequence of timestamps in order.
For example, it may be desirable to expire certain types of email more quickly than others
(say, email internal to an organization which is likely to be delivered faster,
or email with particularly sensitive content). Alternatively, it may be desirable
to expire email not only according to its timestamp but also based on other metadata 
(say, a recipient whose email server is revealed to have been compromised).
Our FFS definition would support expiry on the basis of any attribute which is recorded in the tag:
for example, tags might be of the form $\tau=(t,\delta)$ 
where $t$ is a timestamp and $\delta$ is the delay until expiration,
or alternatively, $\tau=(t,R)$ where $R$ is the recipient.
Moreover, even when tags have the form $\tau=(t,\cdot)$,
our construction still allows for efficiently expiring based solely on timestamp:
that is, expiring ``all tags that begin with $t$.''

\paragraph{FFS in the publicly verifiable timekeeper model}
Recall the \emph{publicly verifiable timekeeper} (PVTK) model, in which 
a reliable timekeeper periodically issues timestamps
whose authenticity is publicly verifiable (e.g., because the authority signs each timestamp).
We show that in this model, expiration of signatures may occur ``automatically'' over time:
that is, it is not necessary for the signer to publish any additional expiry information
in order for the signature to become forgeable after a delay.
In this case, the algorithm $\Expire$ is unnecessary, and $\Forge$ need not take $\eta$ as input.
Our construction of FFS in the PVTK model is given in Section~\ref{sec:timestamping}.

\smallskip

\ifsubmission
\else
  \paragraph{Difference with forward-secure signatures}
Both FFS and FSS yield a system of short-lived secret keys 
all corresponding to one long-lived public key.
However, the definitions differ in two main ways, described below
and depicted in Figure~\ref{fig:fss}.
\begin{enumerate}
  \item Forward-\emph{secure} signatures require that past keys cannot be computed
  from future keys, whereas forward-\emph{forgeable} signatures
  require that future keys cannot be computed from past keys.
  \item Forward-\emph{secure} signatures are specifically designed to protect against compromise of all secret key material,
  so the each key must be derivable based solely on the previous short-lived key.
  Forward-\emph{forgeable} signatures, in contrast, may have persistent ``master secret key''
  material that is used to generate each short-lived key.
\end{enumerate}

\begin{figure}[ht!]
  \centering
      {\footnotesize
      \begin{tikzpicture}
      \tikzstyle{key} = [circle, draw, align=center]
      \tikzstyle{msk} = [rectangle, draw, align=center, minimum height=0.5cm]
      \tikzset{myarr/.style={>=stealth,thick}}
      \tikzset{myarr2/.style={>=stealth,thick,dotted}}
      \tikzset{myarr3/.style={>=stealth,thick,dotted,decoration={markings,mark=at position 0.5 with {
        \draw [red,thick,-,solid]
          ++ (-\StrikeThruDistance,-\StrikeThruDistance)
          -- ( \StrikeThruDistance, \StrikeThruDistance);}
      },
      postaction={decorate},
      }}

      \node[key] (key2) {$sk_1$};
      \node[key] (key3) [right=0.4cm of key2] {$sk_2$};
      \node[key] (key4) [right=0.4cm of key3] {$sk_3$};
      \node[key] (key5) [right=0.4cm of key4] {$sk_4$};
      \node (dots) [right=0.2cm of key5] {$\dots$};
      \node (fss) [left=1.2cm of key2] {Forward-secure:};

      \draw[->,myarr] (key2) to [out=20,in=160] (key3);
      \draw[->,myarr] (key3) to [out=20,in=160] (key4);
      \draw[->,myarr] (key4) to [out=20,in=160] (key5);

      \draw[->,myarr3] (key3) to [out=200,in=340] (key2);
      \draw[->,myarr3] (key4) to [out=200,in=340] (key3);
      \draw[->,myarr3] (key5) to [out=200,in=340] (key4);

      \node[key] (kkey2) [below=0.3cm of key2] {$sk_1$};
      \node[key] (kkey3) [below=0.3cm of key3] {$sk_2$};
      \node[key] (kkey4) [below=0.3cm of key4] {$sk_3$};
      \node[key] (kkey5) [below=0.3cm of key5] {$sk_4$};
      \node[msk] (kkey1) [left=0.4cm of kkey2] {$msk$};
      \node (ddots) [right=0.2cm of kkey5] {$\dots$};
      \node (ffs) [left=1.2cm of kkey2] {Forward-forgeable:};

      \draw[->,myarr3] (kkey2) -- (kkey3);
      \draw[->,myarr3] (kkey3) -- (kkey4);
      \draw[->,myarr3] (kkey4) -- (kkey5);

      \draw[->,myarr] (kkey1) to [out=300,in=220] (kkey2);
      \draw[->,myarr] (kkey1) to [out=300,in=220] (kkey3);
      \draw[->,myarr] (kkey1) to [out=300,in=220] (kkey4);
      \draw[->,myarr] (kkey1) to [out=300,in=220] (kkey5);
      
      \end{tikzpicture}
      \caption{Forward-secure vs. forward-forgeable signatures}\label{fig:fss}
      }
\end{figure}
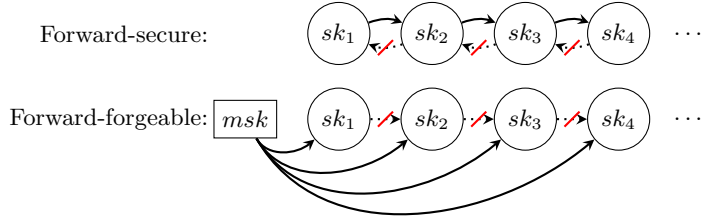
\fi

\subsection{Succinctness}

Next, we define \emph{succinctness} of FFS, a measure of the efficiency of disclosure
in terms of the size of expiry info per tag expired.
\KF uses a construction of FFS based on hierarchical identity-based signatures
(Section~\ref{sec:ffs:construction}), which achieves logarithmic succinctness.

\begin{definition}\label{def:succ}
  Let $z:\NN\to\NN$. Let $S\subset\powset{\Tag}$ be
  a set of sets of tags.
  A forward-forgeable signature scheme
  $\Sig$ is \emph{$(S,z)$-succinct} if
  for any $T\in S$,
  there is a negligible function $\eps$ such that for all $\sec$,
  \begin{align*}
  \Pr_{(vk,sk)\gets\KeyGen(1^\sec)}\Big[
    \big|\Expire(sk,T)\big|\leq z(|T|)
  \Big] \geq 1-\eps(\sec)\ .
  \end{align*}
\end{definition}

\ifsubmission\else
  \begin{remark}
  Definition~\ref{def:succ} is a worst-case definition: it guarantees
  the size of expiry information with overwhelming probability. 
  In certain applications, an average-case definition may be appropriate instead,
  i.e., defining succinctness by bounding the size on \emph{expectation}.
  We use a worst-case definition since it is stronger than an average-case definition,
  and our construction achieves it.
  \end{remark}
\fi

\subsection{FFS Construction from (Hierachical) IBS}\label{sec:ffs:construction}

We first outline a simple FFS construction $\SFFS$ based on identity-based signatures (IBS) \cite{Sha84},
as a stepping stone to our main construction from hierarchical IBS (HIBS).
The next paragraph assumes familiarity with standard IBS terminology;
readers unfamiliar with IBS may skip to the main construction which is explained
formally with explicit syntax definitions.

Let tags in the FFS correspond to identities in the IBS.
$\SFFS.\KeyGen$ outputs IBS master keys. 
The $\SFFS$ signing and verification algorithms for tag $\tau$
respectively invoke the IBS signing and verification algorithms for identity $\tau$.
$\SFFS.\Expire$ outputs the secret key for each input tag $\tau\in T$,
and $\SFFS.\Forge$ uses the appropriate secret key from the expiry information
to invoke the IBS signing algorithm.
The simple solution has linear succinctness. 
By leveraging hierarchical IBS (HIBS), our main construction achieves logarithmic succinctness,
as described next.

\begin{definition}\label{def:HIBS}
A \emph{hierarchical identity-based signature scheme} $\HIBS$ is implicitly
parametrized by message space $\Msg$ and identity space $\Id=\{\Id_\ell\}_{\ell\in\NN}$, and
consists of four algorithms
$$\HIBS=(\Setup,\KeyGen,\Sign,\Verify)$$
with the following syntax.
\begin{itemize}
  \item $\Setup(1^\sec)$ takes as input a security 
  parameter\footnote{Technically, all four algorithms take $1^\sec$ as an input,
  and $\Msg$ and $\Id$ may be parametrized by $\sec$.
  For brevity, we leave this 
  implicit except in $\Setup$.} and outputs a master key pair $(mvk,msk)$.
  \item $\KeyGen(sk_{\vec{id}},id)$ takes as input a secret key $sk_{\vec{id}}$
  for a tuple of identities $\vec{id}=(id_1,\dots,id_\ell)\in\Id_1\times\dots\times\Id_\ell$
  and an additional identity $id\in\Id_{\ell+1}$ 
  and outputs a signing key $sk_{\vec{id'}}$ where $\vec{id'}=(id_1,\dots,id_\ell,id)$.
  The tuple may be empty (i.e., $\ell=0$): in this case, $sk_{()}=msk$.
  \item $\Sign(sk_{\vec{id}},m)$ takes as input a signing key $sk_{\vec{id}}$
  and a message $m\in\Msg$, and outputs a signature $\sigma$.
  \item $\Verify(mvk, \vec{id}, m, \sigma)$ takes as input master verification key $mvk$, 
  tuple of identities $\vec{id}$,
  message $m\in\Msg$, and signature $\sigma$, and outputs a single bit indicating
  whether or not $\sigma$ is a valid signature with respect to $mvk$, $\vec{id}$, and $m$.
\end{itemize}
A \emph{depth-$L$ HIBS} is a HIBS where the maximum length of identity tuples is $L$,
i.e., the identity space is $\Id=\{\Id_\ell\}_{\ell\in[L]}$.
\end{definition}

\begin{definition}
  For an identity space $\Id=\{\Id_\ell\}_{\ell\in\NN}$,
  we say $\vec{id}$ is a \emph{level-$\ell$ identity} if $\vec{id}\in\Id_1\times\dots\times\Id_\ell$.
  For any $\ell'>\ell$, let
  $\vec{id}$ is a level-$\ell$ identity and $\vec{id'}$ is a level-$\ell'$ identity.
  We say that $\vec{id'}$ is a \emph{sub-identity} of $\vec{id}$ if $\vec{id}$ is a prefix of $\vec{id'}$.
  If moreover $\ell'=\ell+1$, we say $\vec{id'}$ is a \emph{immediate sub-identity} of $\vec{id}$.
\end{definition}

\paragraph{Deriving subkeys}
Observe that given a master secret key of a HIBS,
it is possible to derive secret keys corresponding to level-$\ell$ identities for any $\ell$,
by running $\KeyGen$ $\ell$ times.
By a similar procedure, given any secret key corresponding to a level-$\ell$ identity $\vec{id}$,
it is possible to derive any ``subkeys'' thereof, i.e., 
secret keys for sub-identities of $\vec{id}$.
For our construction, it is useful to name this procedure and refer to it directly:
we define $\HIBS.\KeyGenFromRoot$ in Algorithm~\ref{alg:KeyGenFromRoot}.
We write the randomness $\rho_1,\dots,\rho_\ell$ of $\HIBS.\KeyGenFromRoot$ explicitly.

\begin{algorithm}\caption{$\HIBS.\KeyGenFromRoot$}\label{alg:KeyGenFromRoot}
\begin{algorithmic}
  \Input $sk,\ell,\vec{id}=(id_1,\dots,id_{\ell'})$ \Comment{\textit{Require:} $\ell\leq\ell'$}
  \Rand $\rho_1,\dots,\rho_{\ell'}$
  \For{$j=\ell+1,\dots,\ell'$}
    \State $sk\gets\HIBS.\KeyGen(sk,id_j;\rho_j)$
  \EndFor
  \State \textbf{return} $sk$
\end{algorithmic}
\end{algorithm}

\paragraph{HIBS security requirements}
Due to space constraints, Definition~\ref{def:HIBS} gives only the syntax and not the security guarantees
of a HIBS. Informally, an HIBS must satisfy the following security guarantees.
For a formal security definition, see, e.g., \cite{GS02}.
\begin{itemize}
  \item \emph{Correctness:}
  For any identity tuple $\vec{id}$, an honestly produced signature w.r.t $\vec{id}$ must verify as valid w.r.t. $\vec{id}$.
  \item \emph{Unforgeability:}
  For any PPT adversary $\cA$ that has access to a $\KeyGen(msk,\cdot,\cdot)$ oracle,
  the probability that $\cA$ outputs a valid signature w.r.t. an identity $\vec{id}\notin Q$
  must be negligible, where $Q$ is the set of all sub-identities of identities $\cA$ has queried to the oracle.
\end{itemize}

\paragraph{Succinctly representing expiry information}
\ifsubmission\else \begin{algorithm}\caption{$\Compress$}\label{alg:MinCover}
\begin{algorithmic}
  \Input $\Id=\{\Id_\ell\}_{\ell\in[L]},T\subseteq\Id_1\times\dots\times\Id_\ell$
  \State \textit{// Remove redundant sub-identities}
  \ForAll{$\tau\in T$}
    \If{$\exists\tau'\in T$ s.t. $\tau'$ is a prefix of $\tau$}
      \State $T = T\setminus\{\tau\}$
    \EndIf
  \EndFor
  \State \textit{// Replace identities with prefix identities where possible}
  \For{$\ell=L-1,\dots,1$}
    \ForAll{$\vec{\tau}=(\tau_1,\dots,\tau_\ell)\in\Id_1\times\dots\times\Id_\ell$}
      \State \textit{// $X$ represents all level-$(\ell+1)$ sub-identities of $\vec{\tau}$}
      \State $X=\{(\tau_1,\dots,\tau_\ell,\tau')\}_{\tau'\in\Id_{\ell+1}}$
      \If{$X\subseteq T$}
        \State $T = T\setminus X$
        \State $T = T\cup\{(\tau_1,\dots,\tau_\ell)\}$
      \EndIf
    \EndFor
  \EndFor
  \State \textbf{return} $sk$
\end{algorithmic}
\end{algorithm} \fi
Given any set $T$ of tuples of identities, the simplest way to make signatures with respect to $T$ forgeable
would be release the secret key corresponding to each $\vec{id}\in T$, much as in $\SFFS$:
\begin{equation}\label{eqn:naive-rep}
  \eta=\left\{sk_{\vec{id}} = \HIBS.\KeyGenFromRoot(msk,0,\vec{id})\right\}_{\vec{id}\in T}\ .
\end{equation}
However, leveraging the hierarchical nature of HIBS,
$\eta$ can often be represented more succinctly than \eqref{eqn:naive-rep}.
Based on the fact that Algorithm~\ref{alg:KeyGenFromRoot} allows the derivation of any subkey,
we make two optimizations.
First, before computing \eqref{eqn:naive-rep},
we can delete from $T$ any $\vec{id}\in T$ that is a sub-identity of some $\vec{id'}\in T$.
Secondly, if there is any $\vec{id'}=(id_1,\dots,id_\ell)\in\Id_1\times\dots\times\Id_\ell$ such that
every immediate subkey of $\vec{id'}$ is in $T$, i.e.,
$$\forall id_{\ell+1}\in\Id_{\ell+1}, ~ (id_1,\dots,id_\ell,id_{\ell+1})\in T\ ,$$
then all sub-identities of $id'$ can be removed from $T$ and replaced by $id'$ before computing \eqref{eqn:naive-rep}.
Such replacement is permissible only when \emph{every} possible subkey of $id'$ is derivable from $T$:
otherwise, adding $id'$ to $T$ would implicate additional subkeys
beyond those originally in $T$. 

These two optimizations lead us to an algorithm $\Compress$,
which takes as input the identity space of a HIBS and a set of identity tuples $T$,
and outputs a (weakly) smaller set of identity tuples $T'$
such that knowledge of the secret keys corresponding to $T'$ suffices to produce
valid signatures with respect to exactly the identity tuples in $T$.
Next, we describe the operation of $\Compress$ using a tree-based
representation of identity tuples.

\paragraph{Tree representation}
It is convenient to think of identity tuples represented graphically in a tree.
A node at depth $\ell$ represents a tuple of $\ell$ identities (considering the root node to be at depth $0$).
The set of all depth-$\ell$ nodes corresponds to the set of all $\ell$-tuples of identities.
The branching factor at level $\ell$ is $|\Id_{\ell+1}|$.
Given a secret key for a particular node (i.e., identity tuple),
the secret keys of all its descendant nodes are easily computable using $\HIBS.\KeyGenFromRoot$.
(The secret key for the root node is the master secret key.)
In this language, $\Compress$ simply takes a set $T$ of nodes
and returns the smallest set $T'$ of nodes such that (1) all nodes in $T$ are descendants of some node in $T'$
and (2) no node not in $T$ is a descendant of any node in $T'$.
\ifsubmission\else Figure~\ref{fig:SuccExp} gives an illustration of the $\Compress$ algorithm
  on a small example tree.

  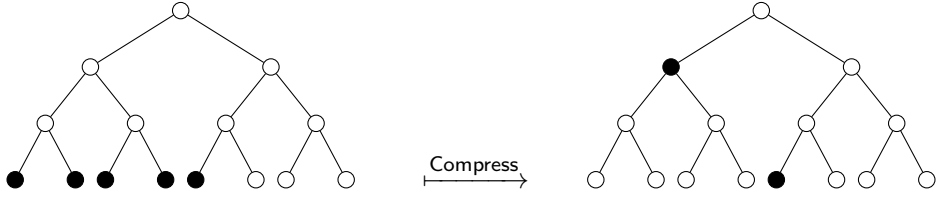
\begin{figure}[!ht]
    \centering
    \ifsubmission
      \begin{tikzpicture}[level/.style={level distance=5mm,sibling distance=16mm/#1}]
    \else
      \begin{tikzpicture}[level/.style={level distance=5mm,sibling distance=16mm/#1}, scale=1.5]
    \fi
      \node [circle,draw,scale=0.6]  {}
        child {node [circle,draw,scale=0.6]  {}
          child {node  [circle,draw,scale=0.6] {}
                      child {node [circle,draw,fill=black,scale=0.6]  {}}
                      child {node [circle,draw,fill=black,scale=0.6]  {}}
                  }
          child {node  [circle,draw,scale=0.6] {}
                      child {node [circle,draw,fill=black,scale=0.6]  {}}
                      child {node [circle,draw,fill=black,scale=0.6]  {}}
                  }
        }
        child {node [circle,draw,scale=0.6]  {}
          child {node  [circle,draw,scale=0.6] {}
                      child {node [circle,draw,fill=black,scale=0.6]  {}}
                      child {node [circle,draw,scale=0.6]  {}}
                  }
          child {node  [circle,draw,scale=0.6] {}
                      child {node [circle,draw,scale=0.6]  {}}
                      child {node [circle,draw,scale=0.6]  {}}
                  }
          };
    \end{tikzpicture}
    \ifsubmission{~}\else{\qquad}\fi$\xmapsto{\Compress}$\ifsubmission{~}\else{\qquad}\fi
    \ifsubmission
      \begin{tikzpicture}[level/.style={level distance=5mm,sibling distance=16mm/#1}]
    \else
      \begin{tikzpicture}[level/.style={level distance=5mm,sibling distance=16mm/#1}, scale=1.5]
    \fi
      \node [circle,draw,scale=0.6]  {}
        child {node [circle,draw,fill=black,scale=0.6]  {}
          child {node  [circle,draw,scale=0.6] {}
                      child {node [circle,draw,scale=0.6]  {}}
                      child {node [circle,draw,scale=0.6]  {}}
                  }
          child {node  [circle,draw,scale=0.6] {}
                      child {node [circle,draw,scale=0.6]  {}}
                      child {node [circle,draw,scale=0.6]  {}}
                  }
        }
        child {node [circle,draw,scale=0.6]  {}
          child {node  [circle,draw,scale=0.6] {}
                      child {node [circle,draw,fill=black,scale=0.6]  {}}
                      child {node [circle,draw,scale=0.6]  {}}
                  }
          child {node  [circle,draw,scale=0.6] {}
                      child {node [circle,draw,scale=0.6]  {}}
                      child {node [circle,draw,scale=0.6]  {}}
                  }
          };
    \end{tikzpicture}
    \caption{Example application of $\Compress$}
    \label{fig:SuccExp}
  \end{figure}

  \bigskip
\fi

\smallskip

Our construction of FFS based on HIBS follows.
\ifsubmission
\else
  It makes use of Algorithms~\ref{alg:KeyGenFromRoot} and \ref{alg:MinCover},
  which were just defined.
\fi

\begin{construction}\label{constr:full}
Let $\HIBS$ be a depth-$L$ HIBS\footnote{The depth need not be finite,
but we consider finite $L$ for simplicity.} 
with message space $\Msg$ and identity space $\Id=\{\Id_\ell\}_{\ell\in[L]}$.
Let $\RO$ be a random oracle,\footnote{%
The construction is presented in the random oracle model for simplicity,
but does not \emph{require} a random oracle:
the random oracle can be replaced straightforwardly by a pseudorandom function (PRF)
where the PRF key is made part of the HIBS secret key.
} and for any tuple $\vec{\tau}=(\tau_1,\dots,\tau_\ell)$, 
let $\vec{\RO}(\vec{\tau})=(\RO(\tau_1),\dots,\RO(\tau_\ell))$.
For $\ell\in[L]$, define $\Tag_\ell=\Id_1\times\dots\times\Id_\ell$.
We construct a FFS
$\Sig$ with message space $\Msg$ and tag space $\Tag=\bigcup_{\ell\in[L]}\Tag_\ell$, as follows.
\begin{itemize}
  \item $\Sig.\KeyGen(1^\sec)$: output $(vk,sk)\gets\HIBS.\Setup(1^\sec)$.
  \item $\Sig.\Sign(sk,\vec{\tau}=(\tau_1,\dots,\tau_\ell),m)$: let 
  $$sk_{\vec{\tau}}=\HIBS.\KeyGenFromRoot(sk,0,\vec{\tau};\vec{\RO}(\vec{\tau}))$$ 
  and output $\sigma\gets\HIBS.\Sign(sk_{\vec{\tau}},m)$.
  \item $\Sig.\Verify(vk,\vec{\tau},m,\sigma)$: \mbox{output $b\gets\HIBS.\Verify(vk,\vec{\tau},m,\sigma)$.}
  \item $\Sig.\Expire(sk,T)$: let $T'=\Compress(\Id,T)$; output
    $$\eta=\left\{(\vec{\tau},sk_{\vec{\tau}}):sk_{\vec{\tau}}=\HIBS.\KeyGenFromRoot(sk,0,\vec{\tau};\vec{\RO}(\vec{\tau}))\right\}_{\tau\in T'}\ .$$
  \item $\Sig.\Forge(\eta,\tau,m)$: if there exists $sk_{\tau'}$ such that
    $(\tau',sk_{\tau'})\in\eta$ and $\tau'$ is a prefix of $\tau$, let $\ell$ be the length of $\tau'$, let
    $$sk_{\vec{\tau}}=\HIBS.\KeyGenFromRoot(sk_{\tau'},\ell,\vec{\tau};\vec{\RO}(\vec{\tau}))$$ 
    and output $\sigma\gets\HIBS.\Sign(sk_{\vec{\tau}},m)$; otherwise, output $\bot$.
\end{itemize}
\end{construction}

\begin{theorem}\label{thm:ffs}
If $\HIBS$ is a secure HIBS,
Construction~\ref{constr:full} instantiated with $\HIBS$ is a 
\ifsubmission\else{forward-forgeable signature scheme}\fi.
\end{theorem}
\ifsubmission\else \begin{proof}
Correctness and unforgeability of Construction~\ref{constr:full} follow directly from
correctness and unforgeability of the underlying HIBS.
The FFS requirement of \emph{forgeability on expiry} moreover follows
from the correctness requirement of the HIBS:
the $\Forge$ algorithm of Construction~\ref{constr:full}
invokes $\HIBS.\Sign$ using a secret key $sk_{\vec{\tau}}$ which
is guaranteed (by construction of $\HIBS.\KeyGenFromRoot$) 
to be the secret key corresponding to identity tuple $\tau$.
The validity of $\HIBS.\Sign$ invoked on a valid secret key corresponding to identity tuple $\tau$
is guaranteed by the correctness of the HIBS.
\end{proof} \fi

Construction~\ref{constr:full} achieves logarithmic succinctness,
as stated next. 

\begin{lemma}[Logarithmic succinctness]\label{thm:succ}
Let $\HIBS$ be a depth-$L$ HIBS with message space $\Msg$ and identity space $\Id=\{\Id_\ell\}_{\ell\in[L]}$.
For each $\ell\in L$, let $\preceq_\ell$ be a total order on $\Id_\ell$.
Let $\Tag_L=\Id_1\times\dots\times\Id_L$.
For $i\in|\Tag_L|$, let $\tau_i$ denote the $i$th element of $\Tag_L$ in
the lexicographic order induced by $\{\preceq_\ell\}_{\ell\in L}$.
Construction~\ref{constr:full} instantiated with $\HIBS$ is
$(S,2z)$-succinct and also $(S_1,z)$-succinct, where
$$S=\big\{\{\tau_i\}_{j'\leq i\leq j}:j,j'\in[|\Tag_L|]\big\} \quad\mbox{and}\quad 
S_1=\big\{\{\tau_i\}_{1\leq i\leq j}:j\in[|\Tag_L|]\big\}\ .$$
and $z(\cdot)=B\cdot\log_B(\cdot)$ where $B=\max_{\ell\in L}\{|\Id_\ell|\}$.
We assume $B$ is constant.
\end{lemma}
\begin{proof}
Fix any $j,j'\in[|\Tag_L|]$ and any set $T=\{\tau_i\}_{j\leq i\leq j'}$.
By the definition of succinctness, 
it suffices to show that the output of $\Compress$ on $T$,
is a set of nodes of size at most $2B\cdot\log(|T|)$.

For any identity tuple $\iota\in\Id$,
let $\mathrm{Sub}_\iota$ be the set of all level-$L$ identities of which $\iota$ is a prefix.
We say that $T$ \emph{covers} $\iota$ if $\mathrm{Sub}_\iota\subseteq T$.
Let ${\rm Cover}_T$ be the set of all identities covered by $T$.
We say $T$ \emph{subsumes} an identity $\iota$ 
if $\iota$ is a descendant of some $\iota'\in{\rm Cover}_T$ such that $\iota'\neq\iota$.
By construction of Algorithm~\ref{alg:MinCover} ($\Compress$),
any identity \emph{subsumed} by $T$ will not be in the output set of $\Compress(T)$
(specifically, it will be removed in the innermost for-loop of Algorithm~\ref{alg:MinCover}).

For any $\ell\in[L]$, consider any consecutive sequence of $s$ level-$\ell$ identities covered by $T$.
By definition of lexicographic ordering, 
there are fewer than $2B$ level-$\ell$ identities in the sequence that are not subsumed by $T$
(these identities will be at the beginning and/or end of the sequence).
Moreover, if the sequence begins at the smallest level-$\ell$ identity in the lexicographic order,
then there are fewer than $B$ identities in the sequence that are not subsumed by $T$
(these identities will be at the end of the sequence).

It follows that for each $\ell\in[L]$,
there are fewer than $2B$ level-$\ell$ identities in the output set of $\Compress(T)$.
The number of levels is $L\leq\log_B(|\Tag_L|)\leq\log_B(T)$.
Therefore, the size of the output set is at most $2B\log_B(T)$.
Moreover, if $T\in S_1$, then the output set size is at most $B\log_B(T)$.
\end{proof}
\section{Our Protocol Proposals}\label{sec:protocols}

\ifsubmission\else
  Sections~\ref{sec:protocols:kf} and \ref{sec:protocols:tf} 
  respectively describe our
  two proposed systems \KF and \TF.
\fi

\subsection{\KF}\label{sec:protocols:kf}

\KF consists of three components:
\begin{enumerate}
  \item Replace the digital signature scheme used in DKIM (currently, RSA)
    with a succinct forward-forgeable signature scheme.
    Details are in Section~\ref{sec:protocols:kf:hibs-details}.
  \item Email servers periodically publish expiry information. Details are in Section~\ref{sec:protocols:kf:publish}.
  \item A \emph{forge-on-request protocol}. Details are in Section~\ref{sec:protocols:kf:burn}.
\end{enumerate}

Note that the first two components already achieve
\emph{delayed universal forgeability} (discussed in Section~\ref{sec:intro:keyideas}).

\subsubsection{FFS configuration for \BasicKF}\label{sec:protocols:kf:hibs-details}

\BasicKF uses the succinct FFS given by Construction~\ref{constr:full} instantiated with a HIBS. 
Our implementation uses the proposed by Gentry and Silverberg \cite{GS02}
which is based on their ``BasicHIDE'' hierarchical identity-based encryption scheme,
hereafter denoted by $\FFS$. 
This \cite{GS02} construction has the advantages of simplicity and low computation costs over many other schemes, and allows us to vary the height of the tree post-implementation.

\BasicKF is based on an $L$-level tag structure, corresponding to identity space
$\Id=\{\Id_\ell\}_{\ell\in[L]}$ where the level-$L$ identities
represent 15-minute time chunks spanning a 2-year period.
We use the following intuitive 4-level configuration for the purpose of exposition,
but as discussed Section~\ref{sec:impl},
it is preferable for efficiency purposes to keep $|\Id_\ell|$ equal for all $\ell\in[L]$.
\begin{align*}
  \Id_1 & = \{1,2\} & \mbox{representing a 2-year time span} \\
  \Id_2 & = \{1,\dots,12\} & \mbox{representing months in a year} \\
  \Id_3 & = \{1,\dots,31\} & \mbox{representing days in a month} \\
  \Id_4 & = \{1,\dots,92\} & \mbox{representing 15-minute chunks of a day}
\end{align*}
A tag $\tau=(\y,\m,\d,\c)\in\Id_1\times\Id_2\times\Id_3\times\Id_4$ thus corresponds to a particular 15-minute chunk of time.
The 15-minute chunks are contiguous, consecutive, and disjoint,
so that any given timestamp is contained in exactly one chunk.
For a timestamp $t$, we write $\tau(t)$ to denote the unique 4-tuple tag $(\y,\m,\d,\c)$
that represents a 15-minute chunk of time containing $t$.
We write $t\sqsubset\tau$ if $\tau$ represents a chunk of time containing timestamp $t$.

\BasicKF requires each DKIM signature at time $t$ to be produced 
with respect to a tag representing timestamp $t+\DE$ 
(except when DKIM signatures are produced within the forge-on-request protocol,
as discussed in Section~\ref{sec:protocols:kf:burn}).
The tag is sent alongside the email, and is used for verification at the receiving server.
The signing and verification algorithms are specified more precisely 
in Algorithms~\ref{alg:BasicKF-sign}--\ref{alg:KF-verify}.

\begin{algorithm}\caption{$\sf\BasicKF.\Sign$}\label{alg:BasicKF-sign}
\begin{algorithmic}
  \Input $sk,m,\Delta$
  \State $t = \CurrTime()$
  \State \textbf{return} $\FFS.\Sign(sk,\tau(t+\Delta),m)$
\end{algorithmic}
\end{algorithm}

\begin{algorithm}\caption{$\sf\BasicKF.\Verify$}\label{alg:KF-verify}
\begin{algorithmic}
  \Input $vk,\tau,m,\sigma$
  \State $t = \CurrTime()$
  \State \textbf{return} $t\sqsubset\tau$ \textbf{AND} $\FFS.\Verify(vk,\tau,m,\sigma)$
\end{algorithmic}
\end{algorithm}

\paragraph{Why 15-minute chunks?}
The time period that a leaf node represents --- i.e., a ``chunk'' ---
corresponds to the maximum granularity at which expiry information can be released,
as also discussed in Section~\ref{sec:protocols:kf:publish}.
$\DE$ is a practical lower bound on the chunk size: since 
$\DE$ represents the time for an email to be delivered, publishing 
expiry information more than once per $\DE$ time does not make sense.
15 minutes is a conservative estimate of $\DE$.

\paragraph{Why a 2-year period?}
Rotating signing keys is good practice;
the Messaging, Malware, and Mobile Anti-Abuse Working Group (M3AAWG)
recommends rotation every 6 months \cite{kurt_andersen_m3aawg_2013}.
Having signing keys that explicitly correspond to time periods,
as in \KF, has the convenient side effect of encouraging regular key rotation.

For an FFS, rotation means generating a fresh master key pair.
Thus, the leaves corresponding to any master key need not exceed the rotation time span.
Recognizing that realistically, key rotation times often exceed 6 months,
we run our evaluation assuming a 2-year period.

\paragraph{How many levels?}
We evaluate performance for different hierarchy depths; our evaluation is discussed
in Section~\ref{sec:impl}.
We find that the best value of $L$ depends on a tradeoff between
computation time and succinctness of expiry information;
the optimal value may thus depend on the specific application.

\subsubsection{Publishing expiry information}\label{sec:protocols:kf:publish}

\KF requires email servers to publish expiry information at regular intervals.
A natural option is to publish expiry information every 15 minutes:
that is, to publish the expiry information corresponding to each chunk $\c$
at the end of the time period that $\c$ represents.

Publishing every 15 minutes yields the finest granularity of expiry possible under the basic four-level tag structure.
Based on an email server's preference, it could release information at longer intervals (e.g., days).
In case of an attack, an adversary would be able to convince third parties
of the authenticity of all emails in the current interval (e.g., the current day),
so the risk-averse approach is to prefer shorter intervals.
If an email server prefers to use intervals even shorter than 15 minutes,
the basic tag structure could easily be adapted for this.

\paragraph{Flexible expiry policies}
The basic tag structures described above can be augmented easily.
For example, an additional level $\Id_\ell$ might
represent the ``sensitivity'' of an email: this would allow
expiry information to be released
with respect to highly sensitive emails more quickly.
Alternatively, it may be desirable that certain emails expire more slowly or not at all
(e.g., an email clearing an employee for a certain privilege for a year,
or emails in the course of an important contract negotiation),
and that is expressible in the sensitivity paradigm too:
sensitivity can be expressed as the desired delay until expiry.

\KF is very configurable: beyond the first four levels, 
different email servers' policies need not be consistent.
The verification algorithm refers only to the first four levels of the tag
(when checking $t\sqsubset\tau$),
so a sending server can add more levels beyond the basic four
without sacrificing compatibility,
to match its desired expiry policy.

\subsubsection{Forge-on-request protocol}\label{sec:protocols:kf:burn}

The final component of \KF is a protocol by which
email servers accept real-time requests for specified email content to be sent to the requester.
Any client with the capability of sending outgoing mail
can thus initiate a request.

The recipient of the requested email is required to be the requester:
this is essential
so the expiry request protocol cannot be used for spoofing and spam.
To enforce this,
receiving servers only respond to requests signed by the requester's server,
and address the response to the originating client.
The request content that is signed by the requester's server 
includes the identity of the originating client.\footnote{We note that a malicious server could, 
in principle, unauthorizedly sign requests for
any client account it controls. This is outside our threat model,
and is equally possible for outgoing mail under DKIM; see
``Client-server trust'' under Section~\ref{sec:model:tm} for further discussion.}

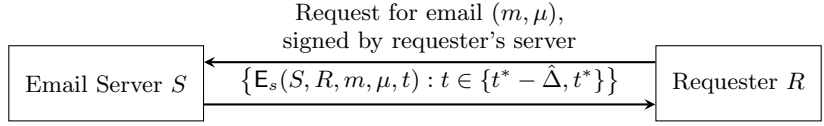
\begin{figure}[ht!]
    \centering
    {\footnotesize
    \begin{tikzpicture}
    \tikzstyle{server} = [rectangle, draw, minimum width=1cm, minimum height=1cm, align=center]
    \tikzset{myarr/.style={>=stealth,thick}}
    
    \node[server] (server) {~Email Server $S$~};
    \ifsubmission
\else
      \node[server] (requester) [right=6cm of server] {~Requester $R$~};
    \fi

    \draw[<-,myarr] ([yshift=8pt]server.east) -- node [above,midway,align=center] {Request for email $(m,\mu)$, \\ signed by requester's server} ([yshift=8pt]requester.west);
    \draw[->,myarr] ([yshift=-8pt]server.east) -- node [above,midway] {$\big\{\Email_s(S,R,m,\mu,t):t\in\{t^*-\DE,t^*\}\big\}$} ([yshift=-8pt]requester.west);
    
    \end{tikzpicture}
    \caption{Forge-on-request protocol \\
    \textit{\small The notation $\Email_s(S,R,m,\mu,t)$ is as defined in Section~\ref{sec:sysreq:def}, and denotes
    an email $(m,\mu)$ from $S$ to $R$ at time $t$. $t^*$ denotes the current timestamp when the reply is sent. Protocol messages are sent via email.}}\label{fig:burn}
    }
\end{figure}

\ifsubmission\else
  Note that a forge-on-request protocol achieves a stronger guarantee
  than the use of ring signatures 
  (i.e., signing with respect to both the sending and receiving servers' public keys).
  A forge-on-request protocol enables any recipient with the ability to send mail from
  an email server to forge mail from any sender to herself.
  The ring signature approach enables her to do this only if she has the ability
  to sign fraudulent mail with the receiving server's secret key.
\fi

\begin{theorem}\label{thm:ffs1}
\KF is non-attributable for recipients (Definition~\ref{def:NA1})
and $\DE$-universally non-attributable (Definition~\ref{def:NA2}).
\end{theorem}
\ifsubmission
\else
  \begin{proof}
Follows directly from Lemmata~\ref{lem:ffs1} and \ref{lem:ffs2}.
\end{proof}

\begin{lemma}\label{lem:ffs1}
\KF is non-attributable for recipients (Definition~\ref{def:NA1}).
\end{lemma}
\begin{proof}
Recall from Definition~\ref{def:NA1} that we must show that
there is a PPT simulator $\cS$ such that for any sender $S$ and recipient $R$,
for any email message $m$ and metadata $\mu$,
\begin{equation}\label{eqn:pf1}
\mathsf{\KF}_{sk}(S,R,m,\mu,t) \cind \cS^{R}(S,m,\mu) \ ,
\end{equation}
where $sk$ is the (master) secret key of $S$, $t$ is the time at which $\cS$ is invoked, and
the superscript $R$ denotes that $\cS$ has access to the recipient's email server.
We construct $\cS$ in Algorithm~\ref{alg:pf1}.

\begin{algorithm}\caption{Simulator $\cS$ for recipient non-attributability}\label{alg:pf1}
\begin{algorithmic}
  \Input $S,m,\mu$
  \State $t = \CurrTime()$
  \State \textbf{send} forge request $(m,\mu)$ to $S$
  \State \textbf{receive} answer $\{e_0,e_1\}$
  \State \textbf{parse} $e_0,e_1$ as emails w.r.t. tags $\tau_0,\tau_1$ respectively
  \If {$\tau(t)=\tau_0$} 
    \textbf{return} $e_0$
  \ElsIf {$\tau(t)=\tau_1$}
    \textbf{return} $e_1$
  \EndIf
\end{algorithmic}
\end{algorithm}

By definition of $\DE$, we know $S$ received the request at some time $t'\leq t+\DE$.
Thus, by construction of \KF,
the emails $e,e'$ must be signed with respect to the tags $\tau(t'-\DE),\tau(t')$ (say, respectively).
It follows that $\tau(t)=\tau(t'-\DE)$ or $\tau(t)=\tau(t')$.
Therefore, at least one of the if-conditions in Algorithm~\ref{alg:pf1} must be satisfied,
and $\cS$ always produces an output.
By construction of the if-statements and the forge-on-request protocol,
the output of $\cS$ is an email signed by $S$ for a tag corresponding to timestamp $t$,
as required by \eqref{eqn:pf1}. Indeed, we achieve equality of distributions,
rather than just indistinguishability.
\end{proof}

\begin{lemma}\label{lem:ffs2}
\KF is $\DE$-universally non-attributable (Definition~\ref{def:NA2}).
\end{lemma}
\begin{proof}
Recall from Definition~\ref{def:NA2} that we must show that
there is a PPT simulator $\cS$ such that for any sender $S$ 
(with secret key $sk$) and recipient $R$,
for any email message $m$, metadata $\mu$, and timestamp $t$,
the following holds at any time $\geq t+\DE$:
\begin{equation}\label{eqn:pf2}
{\sf\KF}_{sk}(S,R,m,\mu,t) \cind \cS(S,R,m,\mu,t) \ .
\end{equation}

Let ${\sf\KF}^*_{sk^*}$ be defined as identical to ${\sf\KF}_{sk}$
except that whenever ${\sf\KF}_{sk}$ invokes ${\sf\KF}.\Sign(sk,\cdot)$,
${\sf\KF}^*_{sk^*}$ instead invokes $\HIBS.\Sign(sk^*,\cdot)$.
Using this notation, we construct $\cS$ in Algorithm~\ref{alg:pf2}.

\begin{algorithm}\caption{Simulator $\cS$ for $\DE$-strong non-attributability}\label{alg:pf2}
\begin{algorithmic}
  \Input $S,R,m,\mu,t$
  \State \textbf{retrieve} published expiry information $\eta$ for $S$
  \ForAll {$(\vec{\tau},sk_{\vec{\tau}})\in\eta$}
    \If {$t\sqsubset\vec{\tau}$}
      \State $sk^*\gets\HIBS.\KeyGenFromRoot(sk_{\vec{\tau}},\ell,\tau(t);\vec{\RO}(\vec{\tau}))$
      \State \textbf{return} ${\sf\KF}^*_{sk^*}(S,R,m,\mu,t)$
    \EndIf
  \EndFor
\end{algorithmic}
\end{algorithm}

Since $\cS$ is invoked at a time $\geq t+\DE$,
and $\KF$ prescribes publication of expiry information at the end of each time chunk of duration $\DE$,
it holds that the expiry information $\eta$ retrieved by $\cS$ includes expiry information
with respect to time $t$.
Therefore, the if-condition in Algorithm~\ref{alg:pf2} will be satisfied for at least one element of $\eta$.\footnote{%
In fact, it will be satisfied for exactly one element of $\eta$, 
by construction of $\Compress$ which ensures that no timestamp is represented by more than one element.
}

Recall that ${\sf\KF}.\Sign$ invokes $\FFS.\Sign$,
which in turn invokes $\HIBS.\Sign$.
By definition, if $t\sqsubset\vec{\tau}$ and $\eta$ is expiry information with respect to $sk$,
then $sk^*$ as computed in Algorithm~\ref{alg:pf2} is the same key 
used to invoke $\HIBS.\Sign$ (within $\FFS.\Sign$, which is within $\sf\KF.\Sign$) at time $t$.
Therefore, the output distributions of ${\sf\KF}^*_{sk^*}$ and ${\sf\KF}_{sk}$
are identical. It follows that $\cS$ satisfies \eqref{eqn:pf2}.
Again, we achieve equality of distributions, not just indistinguishability.
\end{proof}
\fi


\subsection{\TF}
\label{sec:timestamping}\label{sec:protocols:tf}

The main limitation of the KeyForge approach is that it requires signers to continuously release key material. Once this key material
has been revealed and widely distributed, any party can formulate apparently-valid signatures. 
The need to widely distribute this key material
can pose a practical challenge for users, who must depend on their provider to perform this task reliably. Failures to distribute key material may in practice limit the degree of deniability that a system can offer.

The TimeForge protocol uses a different approach, one that is designed to ensure that {\em any} party can forge signatures after the time period
$\Delta$ has expired, even if the signer fails to release key material. Rather than depending on key release, TimeForge employs a new service that we term a {\em publicly verifiable timekeeper} (PVTK). This service, which can in practice be realized using various extant Internet systems, does not share any secret key material with the signing party. Instead, it is a global service that maintains a monotonically-increasing clock. The essential property of a PVTK is that, at any clock time $t$, any party with access to the service can obtain a publicly verifiable proof $\pi_{t}$ that the current time is at least $t$. Simultaneously, the system ensures that it is infeasible for an attacker to forge such a proof at an earlier period.

Given a PVTK service, the intuition behind the TimeForge protocol is straightforward. Let $M$ be some email message sent at time period $t$. The sender first signs each message using a standard SUF-CMA signature scheme to produce a signature $\sigma$. However, rather than using this signature directly, she instead authenticates the message using a witness indistinguishable and non-interactive proof-of-knowledge (PoK) of the following informal statement: 
\ifsubmission
\else
	\begin{quotation}
	I know a valid sender signature $\sigma$ on $M$\\~~~~\centering{OR}~~~~\\I know a valid PVTK proof $\pi_{t+d}$, for some $d \ge \Delta$.
	\end{quotation}
\fi

Assuming a trustworthy PVTK service, this proof authenticates the message during any time period prior to $t+\Delta$. Once a PVTK proof $\pi_{t+\Delta}$ becomes generally available, the knowledge proof becomes trivial for any party to formulate. The witness indistinguishability property of the proof system ensures that a proof formulated by the valid signer cannot be distinguished from a forgery constructed using a revealed PVTK proof at a later time period.

\newcommand{\keygensig}{{\sf Sig.Keygen}}
\newcommand{\signsig}{{\sf Sig.Sign}}
\newcommand{\verifysig}{{\sf Sig.Verify}}
\newcommand{\setuptk}{{\sf TK.Setup}}
\newcommand{\provetk}{{\sf TK.Prove}}
\newcommand{\verifytk}{{\sf TK.Verify}}
\newcommand{\setuptf}{{\sf TF.Keygen}}
\newcommand{\provetf}{{\sf TF.Sign}}
\newcommand{\verifytf}{{\sf TF.Verify}}
\newcommand{\forgetf}{{\sf TF.Forge}}
\newcommand{\verify}{{\sf Verify}}
\newcommand{\params}{{params}}
\newcommand{\secparam}{{\lambda}}
\newcommand{\sk}{{sk}}
\newcommand{\pk}{{pk}}

\paragraph{Publicly verifiable timekeeping.}
As a building block for our main result, we define a publicly verifiable timekeeping (PVTK) scheme. 

\begin{itemize}
\item $\setuptk(1^\secparam)$  takes an adjustable security parameter $\secparam$ and outputs a set of public parameters $\params$ and a trapdoor $\sk$. 
\item $\provetk(\sk, t)$  takes as input $\sk$ and the current time epoch $t$, and outputs a proof $\pi_t$. 

\item $\verifytk(\params, t, \pi_t)$ on input $\params$, a time period $t$, and the proof $\pi_t$, 
outputs whether $\pi_t$ is valid.
\end{itemize}

\medskip
\noindent 
{\em Correctness and Security.} Correctness for a PVTK scheme is straightforward. For security, we consider an experiment in which an adversary with access to an oracle that provides PVTK proofs for arbitrary time periods $t$, must not be able to produce a valid proof for some time period $t_{\sf max}+\Delta$ (except with negligible probability) where $t_{\sf max}$ is the largest value queried to the oracle during the experiment, and $\Delta > 0$ is a constant chosen as an input to the experiment. If a PVTK scheme satisfies this definition for some $\Delta > 0$, then we refer to it as $\Delta$-PVTK secure. 

In practice, we note that for any $\Delta > 0$, a $\Delta$-PVTK scheme can be realized using an {\sf SUF-CMA}-secure signature scheme, where $\provetk$ and $\verifytk$ are respectively implemented using the signing and verification algorithms on message $t$. We leave the problem of constructing and proving weaker schemes to future work. 

\paragraph{Realizing a PVTK service.} 
A simple PVTK system can be constructed using a single server that maintains a monotonically-increasing clock, and implements the PVTK protocol by signing the current time period  using an SUF-CMA signature scheme. While this solution is conceptually simple, practically deploying it at the required scale is likely to be quite costly. Moreover, such a system would be highly vulnerable to denial of service and network-based attacks.

A better approach would be to construct a PVTK system from {\em existing} Internet services, which have already been deployed at scale.
\ifsubmission
\else
	Next, we outline several proposals.

	\begin{description}
	\item[OCSP servers.] The Online Certificate Status Protocol in its ``stapling'' configuration~\cite{rfc6961} allows TLS servers to obtain a standalone, signed certificate validity message from a Certificate Authority (CA). This message is cryptographically signed by the CA, and contains a timestamp intended to guarantee freshness of the data within a ten day expiry window. OCSP staple generation can therefore be viewed as an organic implementation of a PVTK protocol. To avoid the need for reliance on a single CA's infrastructure, users can define the proof $\pi_{t}$ to comprise {\em multiple} valid staples, {\em e.g.,} one each from any $k$ out of $N$ chosen CAs. These parameters, as well as the CA identities, can be selected as part of the setup algorithm.

	\item [Centralized services: CT and randomness beacons.] The Certificate Transparency protocol consists of a centrally-managed and publicly verifiable log for recording the issuance of certificates~\cite{rfc6962}. Each CT log entry is signed by the log operator ({\em e.g.,} Google). While CT is not intended as a timestamping protocol, signed CT log entries contain timestamps and may be re-purposed to implement a centralized PVTK service without the need to deploy new infrastructure. Along similar lines, NIST operates a randomness beacon~\cite{nistbeacon} that also distributes signed records containing a current timestamp. While any single centralized service may be unreliable or subject to attack, a (fault-tolerant) combination of these extant services can be used to construct a ``composite'' PVTK system. 

	\item[Proof of work blockchains.] A number of cryptocurrencies use {\em proof of work} blockchains to construct an ordered transaction ledger~\cite{nakamoto2012bitcoin,wood2014ethereum}. In these systems, generating each new block involves substantial computational effort by a large network of computers, using parameters tuned to produce new blocks at some chosen average rate. These ledgers can be used as a restricted form of PVTK system, in which $\params$ comprises some initial block header $B_s$, and $\pi_{t}$ comprises a list of block headers $\{B_{s+1}, \dots, B_t\}$ drawn from the blockchain. Verification involves checking each proof of work and the correct construction of the chain. This approach does not produce an exact timekeeping service (due to the probabilistic nature of block intervals, and difficulty adjustments), nor does it guarantee cryptographic unforgeability, because chains can be forged by an attacker who controls a substantial fraction of the network's hash power. However, these threats are unlikely to be a significant problem over the short time intervals used in the TimeForge scheme.\footnote{A related development, the {\em proof-of-stake} blockchain~\cite{ouroboros}, employs signatures by currency holders to authenticate new blocks. These techniques might provide an alternative to PoW blockchains for instances where $\Delta$ is small.}

	\item[VDFs and puzzles.] Cryptographic ``puzzles'' are mathematical problems that require a known (or statistically predictable) number of computational operations to solve. Typical puzzles include cryptocurrency proof of work systems and timelock encryption schemes based on brute-forcing cryptographic keys~\cite{rivest96timelock}. A related primitive, the {\em Verifiable Delay Function}~\cite{bonehVDFs,VDFsPietrzak} creates a {\em sequential} puzzle that requires a precisely-known amount of work to solve, and allows the solver to produce a proof of the solution's correctness. While puzzles and VDFs do not directly allow for the creation of a PVTK system, they enable a related primitive: at time $t$ a sender may generate a puzzle challenge $M$ ({\em e.g.,} the contents of an email message) such that a proof $\pi_{t+\Delta}$ can be found by applying a computational process to $M$ in expected time approximately $\Delta$. 
ß\end{description}
\fi

\paragraph{A basic TimeForge signature scheme.} Using the PVTK primitive above, we can now formally describe our basic TimeForge signature scheme. This initial scheme provides only one of our main deniability guarantees: $\Delta$-universal non-attributability (Definition~\ref{def:NA2}). That is, it allows forgery by any party after a fixed time period. 
\ifsubmission
\else
    Below, we discuss how to add {\em recipient non-attributability} as an additional feature.
\fi

The TimeForge scheme consists of three algorithms: $\setuptf$, $\provetf$ and $\verifytf$, as well as a specialized {\em forgery} algorithm $\forgetf$. We assume the existence of a PVTK scheme with parameters $\params$ and an SUF-CMA signature algorithm ${\sf Sig}$.

\begin{itemize}
\item $\setuptf(1^{\secparam}, \params)$. Run $\keygensig(1^{\secparam})$ to generate $(\pk, \sk)$ and output $PK = (\pk, \params)$, and $SK = \sk$.

\item $\provetf(PK, SK, M, t, \Delta)$. Parse $PK = (\pk, \params)$. On input a message $M$ and a time period $t$, first compute $\sigma \leftarrow \signsig(SK, M \| t \| \Delta)$ and compute the following witness-indistinguishable (WI) non-interactive PoK:\footnote{Here we use Camenisch-Stadler notation, where the witness values are in parentheses $()$ and any remaining values are assumed to be public.}
\begin{eqnarray*}
\Pi = {\sf NIPoK}\{ (\sigma, s, \pi) : \verifysig(\pk, \sigma, M \| t \| \Delta) = 1~\vee\\~(\verifytk(\params, \pi, s) = 1~\wedge~s \ge t+{\Delta}) \}
\end{eqnarray*}
Note that the prover can produce this proof using $(\sigma, \bot, \bot)$ as the witness. Output $\sigma_{\sf tf} = (\Pi, t, \Delta)$.

\item $\verifytf(PK, M, \sigma_{\sf tf})$. Parse $PK = (\pk, \params)$ and $\sigma_{\sf tf} = (\Pi, t, \Delta)$ and verify the proof $\Pi$ with respect to the public values $t, \Delta, \pk, M$. If the proof verifies, output $1$, else output $0$.
\end{itemize}

\noindent
We now describe the forgery algorithm $\forgetf$. This forgery algorithm takes as input a PVTK proof $\pi_s$ for some time period $s \ge t+\Delta$:

\begin{itemize}
\item $\forgetf(PK, M, t, s, \Delta, \pi_{s})$. Parse $PK = (\pk, \params)$ and compute the NIPoK $\Pi$ described in the $\provetf$ algorithm, using the witness $(\bot, s, \pi_s)$. Output $\sigma_{\sf tf} = (\Pi, t, \Delta)$.
\end{itemize} 

\ifsubmission
\else
\newcommand{\smax}{s_{\sf max}}
\newcommand{\adversary}{{\cal A}}
\newcommand{\newadversary}{{\cal B}}
\newcommand{\pkpvtk}{pk_{\sf PVTK}}
\newcommand{\skpvtk}{sk_{\sf PVTK}}

\itparagraph{Defining Security.} Security for TimeForge is defined according to the following experiment. This experiment can be considered a variant of (weak) {\sf UF-CMA} security definition for a signature scheme: an attacker must attempt to forge a TimeForge proof over a message $M$ that she has not previously queried to a signing oracle. To assist in this, the attacker is given access to not one, but two oracles. The first is a signing oracle for the TimeForge signature scheme, and produces valid signatures for tuples of the form $(M', t', \Delta')$. The main difference from the standard {\sf UF-CMA} experiment is the existence of a second oracle that models the PVTK service. To model this, the attacker additionally obtains PVTK parameters $\params$ at the start of the experiment, and may repeatedly query the PVTK oracle on chosen epoch numbers $s$ to obtain PVTK proofs of the form $\pi_{s}$. Let $\smax$ be the largest time period queried to the PVTK oracle at the conclusion of the experiment. We say the attacker {\em wins} iff she outputs a message $M$ and valid TimeForge proof $\sigma_{\sf tf} = (\Pi, t, \Delta)$ such that $\smax < t + \Delta$. where $(M, t, \Delta)$ was not queried to the signing oracle. We say that a TimeForge scheme is unforgeable under {\em chosen timestamp attacks} if $\forall$ p.p.t. attackers $\adversary$, the adversary has at most a negligible advantage in succeeding at the above experiment.

\begin{remark} We note that the definition above does not prevent the attacker from constructing forgeries that are ``outside the expiration period.'' Specifically, an intermediary can intercept a message embedding $(M, t, \Delta)$ where $t + \Delta$ is in the future, and author a new message $M', t', \Delta'$ where $t + \Delta$ has already been proved by the PVTK oracle. This is explicitly allowed by TimeForge; indeed, it is a goal of the system. 
\end{remark}

\begin{theorem}
If the PVTK service is constructed using an {\sf SUF-CMA} signature scheme; the WI proof system is sound (extractable); and the underlying signature scheme used by TimeForge is {\sf SUF-CMA}, then the basic TimeForge scheme defined above is secure under chosen timestamp attacks.
\end{theorem}

\begin{proof}
Our proof proceeds by contradiction. Let $\adversary$ be an attacker that succeeds with non-negligible advantage in the chosen timestamp experiment. We show how to construct a pair of algorithms $\newadversary_{1}, \newadversary_2$ such that that one of the two algorithms (respectively) succeeds with non-negligible advantage in the {\sf SUF-CMA} game against $(1)$ the signature scheme {\sf Sig}, or $(2)$ the underlying PVTK signature scheme. 
We now describe the operation of each algorithm.

{\em An attack on the signature scheme {\sf Sig}.} 
In this strategy we construct $\newadversary_1$, which conducts the {\sf SUF-CMA} experiment for the underlying signature scheme {\sf Sig}. $\newadversary_1$ first obtains a public key $\pk$ from the {\sf SUF-CMA} challenger. It next uses the PVTK signature scheme's key generation algorithm to produce a keypair $(\params, \skpvtk)$ for the PVTK service and sends $PK=(\pk, \params)$ to $\adversary$. 

Each time $\adversary$ queries the PVTK oracle on some timestamp $s$, $\newadversary_1$ implements $\provetk$ by using $\skpvtk$ to sign $s$ and return the signature $\pi_{s}$. Whenever $\adversary$ queries the TimeForge signing oracle on some tuple $(M, t, \Delta)$, $\newadversary_1$ first queries the {\sf SUF-CMA} signing oracle to obtain a signature $\sigma$ on $M \| t \| \Delta$. It then constructs a TimeForge signature by constructing the proof described in the $\provetf$ algorithm using $\sigma$ as the satisfying witness, and returns $\sigma_{\sf tf}$ to $\adversary$.

When $\adversary$ outputs a pair $(M^*, \sigma_{\sf tf}^*)$ that satisfies the win conditions of the experiment, $\newadversary_1$ parses $\sigma_{\sf tf}^*$ to obtain $(\Pi^*, t^*, \Delta^*)$ and runs the extractor on $\Pi^*$ to obtain the witness $(\sigma^*, s^*, \pi^*)$. (If the extractor fails, $\newadversary_1$ aborts.) If $\sigma^*$ is a valid signature on $M^* \| t \| \Delta$, it outputs the pair $(\sigma^*, M^*)$ as an {\sf SUF-CMA} forgery.\footnote{Note that by the restrictions on $\adversary$, $M^*$ must represent a message that has not previously been queried to the {\sf SUF-CMA} oracle.}

{\em An attack on the PVTK scheme.} In this strategy we construct  $\newadversary_2$, which conducts the {\sf SUF-CMA} experiment against the PVTK signature scheme. This algorithm proceeds as in Strategy $1$, except that here we set $\params$ to be the public key obtained from the {\sf SUF-CMA} challenger, and generate the keypair $(\pk, \sk) = \keygensig(1^{\secparam})$. Queries to the PVTK oracle are answered by forwarding $s$ to the {\sf SUF-CMA} oracle and returning the resulting signature as $\pi_{s}$, and queries to the TimeForge oracle are answered honestly by running $\provetf$ with $\sk$ as an input.

As in the previous strategy, if $\adversary$ succeeds in the experiment we run the extractor to obtain $\sigma^*, s^*, \pi^*$. However, in this case $\newadversary_2$ verifies that $\pi^*$ is a valid signature on $s^*$ and, if so, outputs $(\pi^*, s^*)$ as a forgery for the {\sf SUF-CMA} experiment. 

\medskip \noindent
{\bf Analysis.} We first observe that, in the view of $\adversary$, the simulation produced by $\newadversary_1$ and $\newadversary_2$ are distributed identically. Thus $\adversary$ must succeed with identical advantage when interacting with either simulation. Next, we point out that if the WI knowledge extractor must fail with probability at most negligible in the security parameter under the assumption that the WI proof system is sound. Thus both adversaries will abort with at most negligible probability due to extraction error.

From these observations it remains only to show that at least one of the two algorithms above must succeed with non-negligible advantage in the {\sf SUF-CMA} experiment. This is shown by observing that the attacker's output $(\Pi^*, t^*, \Delta^*)$ and the extracted witness  $(\sigma^*, s^*, \pi^*)$ must necessarily satisfy the conditions that $(1)$ for all PVTK queries $s$ made during the experiment, it holds that $s < t^* + \Delta^*$ (by the requirements of the experiment), $(2)$ the message $M^* \| t \| \Delta$ has not been queried to the {\sf SUF-CMA} signing oracle (by the requirements of the experiment), and $(3)$ the following conditions are true (by the soundness of the proof system):
$$\verifysig(\pk, \sigma^*, M^*) = 1~\vee\\~(\verifytk(\params, \pi^*, s^*) = 1~\wedge~s^* \ge t^*+{\Delta^*})$$
If $\adversary$ succeeds in the TimeForge experiment with non-negligible advantage, then the extracted witness must contain a valid signature $\sigma^*$ on a message $M^* \| t \| \Delta$ not queried to the signing oracle, {\bf or} it must contain a valid signature $\pi^*$ on some epoch $s^*$ which was not queried to the PVTK oracle. By these requirements, it must be the case that one of $\newadversary_1$ or $\newadversary_2$ succeeds with non-negligible advantage, which concludes our proof.
\end{proof}

\paragraph{Realizing the TimeForge proof system.}
TimeForge can be realized using a variety of WI and ZK proof systems, combined with efficient SUF-CMA signature schemes. For example, a number of pairing-based signature schemes~\cite{Belenkiy2009,camenisch2004signature,Belenkiy2008} admit efficient proofs of knowledge of a signature using simple Schnorr-style proofs~\cite{AHR05}. More recent proving systems such as Bulletproofs~\cite{bulletproofs} and zkSNARKs~(e.g., \cite{pghr13,groth16}) admit succinct proofs of statements involving arbitrary arithmetic circuits and discrete-log relationships. Using the latter schemes ensures short proofs, in the range of several hundred bytes, in some cases with a small, constant verification cost. Thus, even complex PVTK proofs such as block header sequences, can potentially be reduced to a succinct and readily-verified TimeForge signature.
    \ifsubmission\else
  \paragraph{Adding recipient non-attributability.} 
\fi
The basic TimeForge construction above does not provide {\em recipient non-attributability} (Definition~\ref{def:NA1}). This property allows a message {recipient} to forge a TimeForge signature prior to the time period $t + \Delta$. The scheme can be altered to achieve \emph{non-attributability for recipient servers} by modifying the algorithms to take as input the recipient's public key $\pk_{\sf recipient}$. The proof statement then becomes:
\begin{eqnarray*}
\Pi = {\sf NIPoK}\{ (\sigma, s, \pi) : \verifysig(\pk, \sigma, M) = 1~\vee\\
\verifysig(\pk_{\sf recipient}, \sigma, M) = 1~\vee\\
~(\verifytk(\params, \pi, s) = 1~\wedge~s \ge t+{\Delta}) \}
\end{eqnarray*}
This new proof statement ensures that the recipient can, at any point, generate a valid TimeForge signature for any message addressed to it, provided that it has the correct signing key.

\begin{remark}\label{rmk:RNA}
\KF and the above-described version of \TF achieve different variants of recipient non-attributability under Definition~\ref{def:NA1}.
The recipient non-attributability of \TF is weaker than \KF in the following sense.
In both \KF and \TF, a recipient client
can forge with the help of the recipient server (and a recipient server can forge by itself).
In \KF, the client needs the server's help to sign an outgoing email from the client itself
and then to receive a reply email --- these are both features routinely offered by email servers. 
However, in \TF, the client needs the server's help to sign an email purportedly authored
by the party whom the client designates as the sender of the forged email.
Since this is not a feature routinely offered by email servers (and to offer it naively would enable spam and spoofing),
it could be difficult for a client to forge without the collusion of a malicious server:
realistically, this could substantially undermine the goal of non-attributability.
An interesting open question is whether \TF's recipient non-attributability can be strengthened.
\end{remark}
\fi

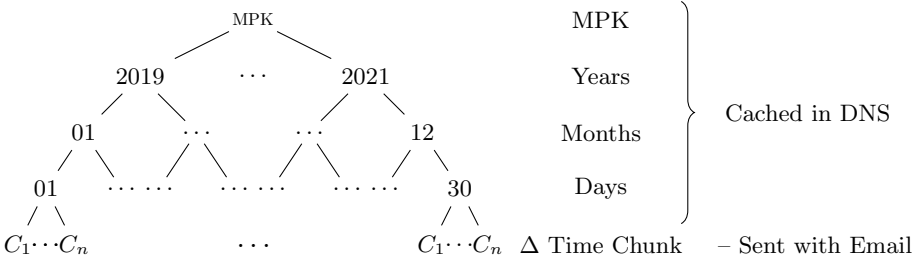
\begin{figure}[!ht]
\centering
\ifsubmission
  \begin{tikzpicture}[level/.style={level distance=5mm,sibling distance=20mm/#1}]
\else
  \begin{tikzpicture}[level/.style={level distance=5mm,sibling distance=20mm/#1},scale=1.5]
\fi
\node [scale=0.6]  {MPK}
child { node (y1)  {\footnotesize 2019}
    child { node  {\footnotesize 01}
        child {node {\footnotesize 01}
            child { node (c12) {\footnotesize $C_{1}$} }
            child { node (cn2) {\footnotesize $C_{n}$} }
        }
        child {node {\footnotesize $\hdots$}}
    }
    child { node {\footnotesize $\hdots$} 
        child {node {\footnotesize $\hdots$}}
        child {node {\footnotesize $\hdots$}}
    }
}
child {node (y2)   {\footnotesize 2021}
    child { node {\footnotesize $\hdots$} 
        child {node {\footnotesize $\hdots$}}
        child {node {\footnotesize $\hdots$}}
    }
    child { node {\footnotesize 12}
        child {node {\footnotesize $\hdots$}}
        child {node {\footnotesize 30}
            child { node (c1) {\footnotesize $C_{1}$} }
            child { node (cn) {\footnotesize $C_{n}$}
                child [grow=right] {node {} edge from parent[draw=none]
                    child [grow=right] {node {\footnotesize  $\Delta$ Time Chunk} edge from parent[draw=none]
                        child [grow=up] {node {\footnotesize Days} edge from parent[draw=none]
                            child [grow=up] {node {\footnotesize Months} edge from parent[draw=none]
                                child [grow=up]  {node {\footnotesize Years} edge from parent[draw=none]
                                    child [grow=up]  {node {\footnotesize MPK } edge from parent[draw=none]
                                    child [grow=right,yshift=.25cm] {node (root) {} edge from parent[draw=none]}
                                    }
                                }
                            }
                        }
                        child [grow=right,yshift=.1cm] {node (dlt) {} edge from parent[draw=none]}
                        child [grow=right,xshift=1.4cm] {node {\footnotesize -- Sent with Email} edge from parent[draw=none]}
                    }
                }
            }
        }
    }
};
\ifsubmission
\else
  \draw [decorate,decoration={brace,amplitude=5pt,raise=9pt}] (root) -- (dlt) node [black,midway,xshift=2cm] {\footnotesize Cached in DNS};
\fi
\path (y1) -- (y2) node [midway] {\footnotesize $\hdots$};
\path (c1) -- (cn) node [midway] {\footnotesize $\hdots$};
\path (c12) -- (cn2) node [midway] {\footnotesize $\hdots$};
\path (c12) -- (cn2) node [midway] {\footnotesize $\hdots$};
\path (cn2) -- (c1) node [midway] {$\hdots$};
\end{tikzpicture}
\caption{DNS / Email Cached Public Parameters}
\label{fig:DNSCache}
\end{figure}

\section{Implementation and Evaluation}
\label{sec:impl:dns-website}

\label{sec:impl}

We have designed and implemented a prototype of \KF suitable for use with common mail (MDA/MSA) servers. The prototype is fully functional and integrates with Postfix and Sendmail. It is designed to be straightforward to extend the system to work with other mail servers, as well as other applications, due to the abstraction of core cryptographic services away from other infrastructure. The entire project consists of roughly 2,000 lines of Go, C, and Python. 
\ifsubmission\else
Our implementation is available to the community open source at \url{https://github.com/keyforgery/KeyForge}.
\fi
There are two main parts to \KF: a key server, and a mail filter that integrates \KF with mail server software, we discuss each in detail below.

\paragraph{Mail Filter:} 
The mail filter is a service that is responsible for ensuring that sent emails are properly formatted, determining if incoming emails are verifiable, and communicating these results the mail server. The filter works by intercepting emails that the mail server sends and receives, adding required metadata to the sent email's header, and requesting cryptographic operations from the key server.

When sending an email, the filter adds the expiry time and any other information necessary to verify the signature to the email's header. The mail filter performs a SHA256 hash of necessary metadata and the message's content and forwards the hash to the key server to sign, and then adds this signature to the header.

On receipt of a signed email, this service confirms that the message metadata exists, signature's hash is equivalent to the hash of the message and metadata to be verified. The filter then sends the signature, the sending domain, and the expiry timestamp to the key server for verification. If any step in the verification fails, the filter alerts the server software, and the message is dropped.

\paragraph{Key Server:}
The key server communicates with the mail filter via JSON RPC~\cite{json-rpc_nodate}, a flexible standard that should allow for easy integration into any other mail system. The key server expires old keys by publishing them on a constantly updated website, though this is an implementation detail that could be augmented in any number of ways.

\ifsubmission
\else
  We use the RELIC\cite{relic-toolkit} cryptographic library's implementation of a 12-degree Barreto-Lynn-Scott curve with $q = 2^{381}$ (BLS12-381). This configuration conservatively yields keys with a 128-bit security level,\cite{aaranha_implementing_2012} which puts our implementation on par with the standard 2048-bit RSA seen in current DKIM implementations.\footnote{We originally considered using a 256-bit Barreto-Naehrig curve, which had a native Go implementation and was close to becoming an IETF standard. Unfortunately, new attacks\cite{kim_extended_2016} lowered the curve's security rating to below the RSA-equivalent 128-bits.} RELIC was chosen due to its support for many pairing friendly curves and computational overhead.
\fi

\subsection{Evaluation}

\paragraph{Bandwidth:} Table~\ref{fig:bandwidth} illustrates \KF's bandwidth costs. 
The second row reflects that email headers and DNS records require that keys be base64 encoded.
Our amortized bandwidth per email is 42\% \emph{smaller} than a DKIM RSA 2048 signature.

A difficulty with most HIBS schemes is that their signatures can be quite large. For example, in our modified Gentry-Silverberg scheme, a signature created from a particular chunk must include public parameters for each node on the path from that chunk to the Master Public Key. A public parameter in our scheme is 97 bytes, so the parameters required for a signature from a tree structured with Y/M/D/Chunk would wind up being $97*4=388$ bytes. The total size also includes the signature point (49 bytes), for a total of 437 bytes.

We further decrease the bandwidth and storage costs our implementation by precomputing the path parameters and storing them in DNS along with the Master Public Key (Figure~\ref{fig:DNSCache}). At verification, the server performs a DNS lookup for the appropriate expiry parameters, and stores them in a local cache. 
This yields a one-time amortized cost of $~$4kb per server per month.

\begin{table*}
\centering
\begin{tabular}{r|ccccc}  
    \toprule
                        & Secret Keys & Public Params & \KF $\sigma$ (w/ cache) & DKIM RSA2048 $\sigma$ \\
    \midrule
    Raw Size (B)    &    64       &  97               &        $49 + 97 = 146$    & 256 \\
    Base64 Size (B) &    88       &  132              &         $201$             & 347 \\
    \bottomrule
\end{tabular}
\caption{Amortized bandwidth costs of various parts of the system, as well as DKIM with RSA}
\label{fig:bandwidth}

\end{table*}

\begin{figure*}[t!]
\begin{subfigure}[t]{0.33\textwidth}
  \centering
  \includegraphics[width=0.95\textwidth]{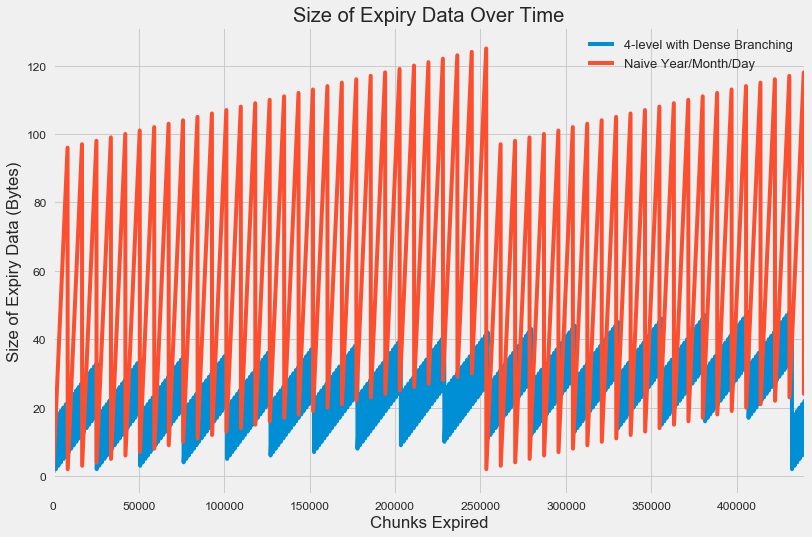}
  \vspace{0in}
  \caption{\scriptsize Size of expiry data given different tree topologies.}
  \label{fig:topo}
\end{subfigure}%
~
\begin{subfigure}[t]{0.33\textwidth}
  \centering
  \includegraphics[width=0.95\textwidth]{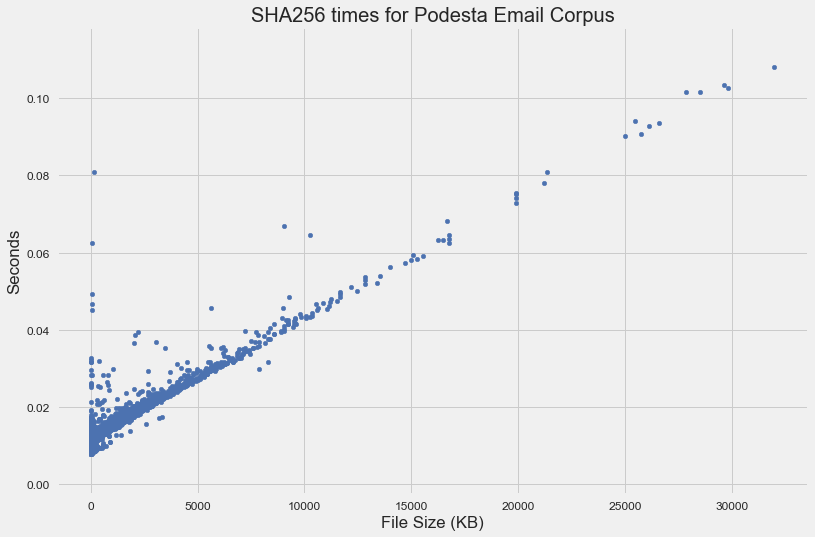}
  \vspace{0in}
  \caption{\scriptsize The time it takes to perform a Sha256 hash of particular messages in the Podesta Email Corpus}
  \label{fig:size}
\end{subfigure}%
~
\begin{subfigure}[t]{0.33\textwidth}
  \centering
  \includegraphics[width=0.95\textwidth]{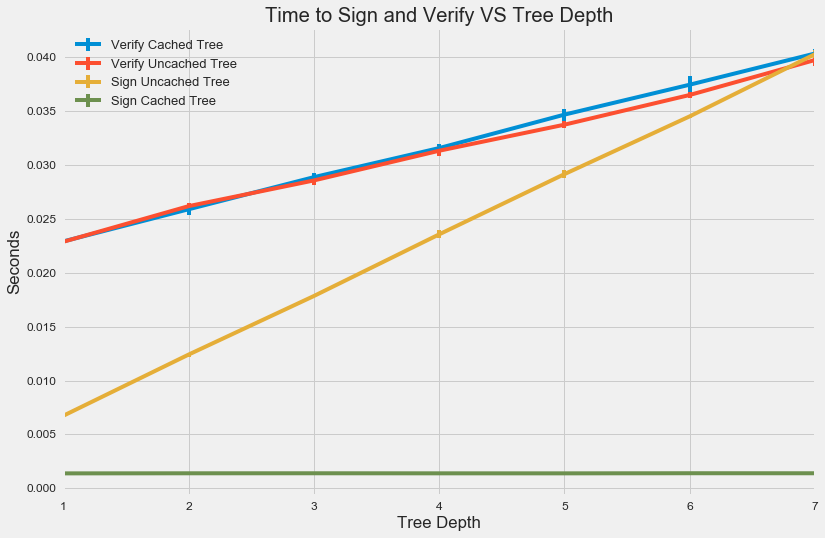}
  \vspace{0in}
  \caption{\scriptsize Time required to sign and verify given whether or not the tree was allowed to pre-cache or is calculating based off of a brand new tree ``cold.''}
  \label{fig:depth}
\end{subfigure}

\smallskip

\caption{Evaluation of Bandwidth and Computation Speed}
\end{figure*}

\paragraph{Computation Costs:} We created a number of microbenchmarks to determine the efficiency of our system, as well as the impact of varying the height of the tree. In particular, evaluated the impact of tree depth on signing and verification time. It is important to note that his is a lower bound on throughput -- all benchmarks were performed on a laptop with power lower than a common server.\footnote{For full details on compilation and machine specifics see appendix \ref{sec:compilation}}

To fully capture the best and worst case for signing and verification, we timed execution given a tree where the public key path from the root must be recalculated each time (as would be common in standing up a brand new server, or initially verifying messages from a new domain) and when the tree is allowed cache such parameters; results are illustrated in Figure~\ref{fig:depth}. We find that signing is largely unaffected by the depth of the tree with caching enabled, which would likely be the case in the majority of settings. Further, we find that increasing the depth of the tree has a linear impact on verification time. 

Cached verification with a depth-4 tree (as in our Y/M/D/Chunk tree) takes an average of 0.031579 seconds per verification operation; one laptop alone could verify roughly 2,735,995 messages per day. Cached signing for the same tree takes 0.001388 seconds per signing operation, which means that same computer could sign roughly 62,247,838 messages per day. 

Signing and verification times are unaffected by plaintext size, since the key server performs operations only on SHA256 hashes of the message. Since hash-and-sign is the common paradigm between DKIM and \KF, both incur these same costs. Figure~\ref{fig:size} illustrates hashing time using OpenSSL's hashing library for different message sizes, using a corpus of real-world leaked emails from Wikileaks' Podesta Email Corpus \cite{wikileaks_wikileaks_2016}.\footnote{Beyond the dramatic irony, we chose the Podesta email corpus because it was distributed intact with attachments, and thus relatively close to a realistic user's email distribution, among publicly available datasets.} We find that the average time to hash for all messages in the corpus is around 0.0102 seconds, indicating that I/O and hashing time are likely as impactful, if not more, than the time it takes to sign and verify.

\paragraph{Optimizing for Expiration Bandwidth:} Succinct expiry information may be more or less of a priority depending on the desired application and associated cost in computation time. For example, succinct expiry information would allow for easy publication (even in DNS) or lightweight attachment to emails. The latter idea has the neat feature of guaranteeing that any person who received emails from this domain possesses expiry information.

While the Y/M/D/Chunk tree configuration is easy to intuit, having an equal branching factor across all tree levels yields a large gain in succinctness. Figure~\ref{fig:topo} illustrates the difference in size between the naive configuration and a depth 4 tree whose branching factor is equal at every level. Table \ref{table:exp} 
shows the average and maximum size of expiry info of various depth trees with an equal branching factor: e.g., the average expiry size for a full 2-year period is 4.5MB, 4KB, or 1.8KB for depths 1, 4, and 7 respectively. An administrator's preferred trade-off between expiry size, depth (i.e., verification time), and validity time of the tree likely depends on the application at hand.

\ifsubmission\else \begin{table}
  \centering
  \caption{Expiry information size}
  \begin{tabular}{|c|c|c|c|c|}
  \hline
  \multirow{3}{*}{$L$} & \multicolumn{4}{c|}{Expiry info size (bytes)} \\
  \cline{2-5}
  & \multicolumn{2}{c|}{1 year} & \multicolumn{2}{c|}{2 years} \\
  \cline{2-5}
  & Avg & Max & Avg & Max \\
  \hline
  1 & 1121248 & 2242496 & 1679814 & 4485056 \\
  2 & 12700 & 25344 & 16934 & 33792 \\
  3 & 3283 & 6464 & 3920 & 7744 \\
  4 & 1787 & 3520 & 2016 & 3968 \\
  5 & 1275 & 2496 & 1408 & 2752 \\
  6 & 1048 & 2048 & 1117 & 2176 \\
  7 & 859 & 1664 & 934 & 1792 \\
  \hline
  \end{tabular}
  \label{table:exp}
\end{table}

 \fi

\paragraph{Optimizing the expiry time:} 
While expiry time is a configurable parameter of \KF (i.e., an administrator may select any appropriate time), there are clear benefits to keeping this time as short as possible.
Indeed, the amount of universal forgeability that a user can enjoy depends directly on the expiry time.

Despite that most emails are received very quickly, the SMTP RFC (RFC 5321) \cite{RFC5321} has a very lax give-up time of 4--5 days. To get a rough idea of how quickly emails tend to be delivered, we computed the time differences from the first Received header to the last, in the Podesta email corpus, and found that of the 48,246 messages with parseable Received timestamps, 99\% (47,349) took less than 15 minutes to be delivered.

We believe that this time can therefore be shortened considerably using any number of tactics. For example, delays are often caused by expected server outages on the receiving end, e.g., by an administrator using the weekend to update email servers. 
This could be resolved by using a DMARC-like DNS record to signal that maintenance is happening and to hold messages until later.

Shortening the expiry time in the face of routing delays is made slightly trickier due to the presence of potentially malicious adversaries; providing flow control similar to that of TCP would certainly be systematically possible, but implementers must be careful to ensure that a malicious MTA is not attempting to keep the message unforgeable for undue time. Setting a hard-cutoff maximum here of one to two days would likely be advisable.

\section{Related Work and Discussion}
\label{sec:related}

A number of works have considered the problem of {\em deniability} in various settings. 

\itparagraph{Deniable authentication.} The cryptographic literature features many works on deniable signatures and authentication, including (but in no way limited to)~\cite{Jakobsson1996designated,dwork98concurrent,RST01,naor2002deniable,DiRaimondo2009}. Indeed, our constructions can be viewed as a practical instantiation of a deniable signature scheme, with tight systems-based requirements on the deniability properties. A much stronger primitive, {\em deniable encryption}~\cite{canetti1997deniable} considers encryption schemes that can survive the later compromise of an encrypting device.
 
\itparagraph{Messaging deniability.} In the context of interactive messaging, the OTR protocol~\cite{borisovOTR} uses MACs and key agreement protocols to ensure that message transcripts can be repudiated by either party to a conversation. Along similar lines, Unger and Grubbs considered deniable key exchanges for new messaging systems~\cite{unger2015deniable}. Adida {\em et al.} proposed ``lightweight email signatures'' for DKIM, based on ring signatures~\cite{AHR05}, but with more limited properties than we consider in this work. Recently, Gunn {\em et al.} noted that trusted hardware systems can be used to circumvent deniability in many protocols~\cite{gunn2018circumventing}. In the theoretical realm, Canetti  {\em et al.}'s recent possibility result addressed a yet stronger deniability guarantee for interactive messaging, where transcripts can be ``explained'' as encrypting any plaintext of the adversary's choice, even when both parties are coerced \cite{CPP18}.

\itparagraph{Message franking.}
Finally, a recent line of works~\cite{dodis2018franking,grubbs17franking} considers the problem of {\em franking} in encrypted messaging systems. These protocols consider a different problem than we address in this work: they allow attribution of messages even after an arbitrary period, but only with the cooperation of a centralized service provider. One such system, based on a secret HMAC key, has been deployed by Facebook~\cite{messengerFranking}.

\itparagraph{Complications from ARC.} As mentioned in Section~\ref{sec:bg}, ARC is experimental and widely unadopted as a standard, but would pose some minor complications for non-attributability if adopted, because it has third-party MTAs sign emails in transit. \KF, though presented as a modification of DKIM, is conceptually straightforward to extend to accommodate ARC, by having third-party signer MTAs to use a FFS for signing, publish expiry information, and offer a forge-on-request protocol, just like the MSA in \KF.

\ifsubmission\else
  \section*{Acknowledgements}
  We are grateful to Jon Callas for his in-depth review of the applicability of our scheme to DKIM, and to Dan Boneh, Daniel J. Weitzner, John Hess, Bradley Sturt, and Stuart Babcock for their feedback. This work was supported in part by the William and Flora Hewlett Foundation grant 2014-1601.
\fi

\bibliographystyle{plain}
\bibliography{refs}

\appendix
\section{Compilation and Evaluation parameters}
\label{sec:compilation}
We performed all benchmarks on a 2017 MacBook Pro, 15-inch, with an Intel 4-core 3.1GHz processor and 16GB of RAM. 

Openssl was compiled with the following flags:
\begin{verbatim}
clang -I. -I.. -I../include  -fPIC -fno-common 
-DOPENSSL_PIC -DOPENSSL_THREADS -D_REENTRANT
-DDSO_DLFCN -DHAVE_DLFCN_H -arch x86_64 -O3 
-DL_ENDIAN -Wall -DOPENSSL_IA32_SSE2 
-DOPENSSL_BN_ASM_MONT
-DOPENSSL_BN_ASM_MONT5 -DOPENSSL_BN_ASM_GF2m 
-DSHA1_ASM -DSHA256_ASM -DSHA512_ASM -DMD5_ASM 
-DAES_ASM -DVPAES_ASM -DBSAES_ASM -DWHIRLPOOL_ASM 
-DGHASH_ASM -DECP_NISTZ256_ASM
\end{verbatim}

RELIC was compiled with the following flags:
\begin{verbatim}
cmake ../ -DALLOC=DYNAMIC -DFP_PRIME=381 
-DARITH=gmp-sec -DWSIZE=64 
-DFP_METHD="INTEG;INTEG;INTEG;MONTY;LOWER;SLIDE"
-DCOMP="-O3 -mtune=native -march=native" 
-DFP_PMERS=off -DFP_QNRES=on 
-DFPX_METHD="INTEG;INTEG;LAZYR"
-DEP_SUPER=off -DPP_METHD="LAZYR;OATEP"
\end{verbatim}

\end{document}